\documentclass[11pt]{article}
\usepackage[margin=1in]{geometry}
\usepackage{amsmath}
\usepackage{amsthm}
\usepackage{amssymb}
\usepackage{amsfonts}
\usepackage{bm}
\usepackage{bbm}
\usepackage{color}
\usepackage{graphicx}
\usepackage{hyperref}
\usepackage{fancyhdr}
\usepackage{graphicx}
\usepackage{slashed}
\usepackage{latexsym,epsfig}
\usepackage{algorithm,algpseudocode}	
\usepackage{enumerate}
\usepackage{mathtools}
\usepackage{physics}
\usepackage[capitalize]{cleveref}
\usepackage{subcaption}
\usepackage{tcolorbox}
\usepackage{authblk}

\definecolor{blue}{rgb}{0,0.2,1}
\definecolor{Pr}{rgb}{0.4,0.3,0.9}

\newcommand{\Var}{\mathrm{Var}}

\theoremstyle{plain}

\newtheorem{problem}{Problem}
\newtheorem{theorem}{Theorem}
\newtheorem{prop}{Proposition}
\newtheorem{lemma}{Lemma}

\newtheorem{defn}{Definition}

\def\be{\begin{eqnarray}}
\def\ee{\end{eqnarray}}

\DeclareMathOperator{\EE}{\mathbb{E}}

\begin{document}
\title{Low depth amplitude estimation without really trying}
\author{Dinh-Long Vu, Bin Cheng, Patrick Rebentrost}
\affil{Centre for Quantum Technologies, National University of Singapore, Singapore, Singapore}
\maketitle
\begin{abstract}
    Standard quantum amplitude estimation algorithms provide quadratic speedup to Monte-Carlo simulations but require a circuit depth that scales as inverse of the estimation error. In view of the shallow depth in near-term devices, the  precision achieved by these algorithms would be low. In this paper we bypass this limitation by performing the classical Monte-Carlo method on the quantum algorithm itself, achieving a higher than classical precision using low-depth circuits. We require the quantum algorithm to be weakly biased in order to avoid error accumulation during this process. Our method is parallel and can be as weakly biased as the constituent algorithm in some cases.
\end{abstract}
\section{Introduction}

Quantum computing can simply be viewed as the task of executing a quantum circuit with measurements on some qubits in the end. The goal of amplitude estimation is to estimate the probability of getting all these measured qubits in a certain configuration, for example the last qubit in the state $|1\rangle$. Such general formulation is perhaps amplitude estimation's biggest strength: aside from unitarity, the circuit can literally be quite anything. In particular, Montanaro observed \cite{Montanaro_2015} that if a random variable can be encoded by the circuit then the above probability is nothing but the corresponding expected value. As a result, amplitude estimation has promising applications in problems where obtaining the expected value, especially with high precision is the main focus. Several use-cases have been demonstrated in finance, including option pricing \cite{PhysRevA.98.022321,Stamatopoulos_2020} and risk management \cite{Woerner_2019,9259208}.

How is the probability estimated? The straightforward way is to repeat the experiment a number of times, record the outputs and take the average. If we want to achieve a precision of $\epsilon$ then this classical Monte-Carlo method requires $O(1/\epsilon^2)$ samples of the original circuit. Such query complexity is definitely not worth the effort of quantum encoding a random variable, which is already a challenging task by itself for most applications. By sequentially concatenating the controlled versions of the original circuit to form a deeper one, it was shown in \cite{Brassard_2002} that the query complexity can be reduced to $O(1/\epsilon)$, which turns out to be the best asymptotic that a quantum algorithm can achieve. One can think of the classical Monte-Carlo method as vertical stacking and the algorithm of \cite{Brassard_2002} as horizontal stacking. There are also other algorithms \cite{Suzuki_2020,Aaronson_2020,Grinko_2021} with optimal query complexity but with simpler circuit architecture. 

Despite the quadratic speedup, all these algorithms come with a catch: the maximum circuit depth (measured in unit of the original circuit depth) is of the same order as the query complexity itself, that is $O(1/\epsilon)$. Indeed, the original algorithm of \cite{Brassard_2002} makes use of a single concatenated circuit so the maximum depth and the query complexity are identical. Other algorithms involve multiple circuits the depths of which form a geometric series therefore the largest depth differs from the query complexity only by an $O(1)$ multiplicative factor. Given the shallow depth limitation of near term quantum devices, the application of amplitude estimation would be limited to low-precision tasks only, which can be considered a waste of potential.

An interesting direction was proposed in \cite{Giurgica_Tiron_2022}: the idea is to trade some of the quadratic speedup for a lower depth. How much do we need to compromise? Denote respectively by $D$ and $N$ the maximum circuit depth and the total query complexity. It was shown in \cite{burchard2019lowerboundsparallelquantum} that their product cannot be smaller than $O(1/\epsilon^2)$. In other words, the optimal speedup is nothing but the maximum circuit depth itself. When $D=O(1/\epsilon)$ we recover the quadratic speedup of standard quantum algorithms whereas taking $D=1$ brings us back to the classical Monte-Carlo method. The term \textit{low-depth} refers to the ability to smoothly interpolate between the two paradigms along the optimal curve $DN=O(1/\epsilon^2)$. Befitting such compelling behaviour, finding low-depth algorithms appears to require somewhat exotic techniques and intricate manipulations. So far two attempts with provable correctness are known. The first one \cite{Giurgica_Tiron_2022} makes use of the Chinese remainder theorem. The idea is to transform the amplitude estimation problem into a problem of finding the remainder when divided by an integer of order $O(\epsilon^{-1})$ chosen in such a way that it can be factorized into co-prime factors of roughly equal size. The remainder modulo each factor can be computed using shallow-depth algorithm and then recombined using the Chinese remainder theorem. The main drawback of this algorithm is that it can only set the maximum depth to be some $n$-th root of $\epsilon^{-1}$, like $\epsilon^{-1/2},\epsilon^{-1/3}$ and so on, as a result of the equal size factorization. Deviating from these discrete values breaks the optimal condition. The second one \cite{Rall_2023} relies on the so-called semi-Pellian polynomials, which contain Chebyshev polynomials as a subclass. Since Chebyshev polynomials correspond to Grover's rotations, the idea is to instead use semi-Pellian polynomials to explore the phase space more efficiently. The low depth requirement forces these  polynomials to have low degrees, thus low predictive power and as a result more samples are required. We observed that there is an error in this algorithm however and in \cref{sec: rall-fuller} we present a modification that fixes the issue. Even with our fix, the construction of these low-degree polynomials is only possible when the maximum circuit depth is at least of order $O(1/a)$ where $a$ is the true value of the amplitude, see \cref{subsec: Rall-Fuller limitation} for more detail on this point. In the extreme case where the true amplitude happens to be $0$, the algorithm of \cite{Rall_2023} ceases completely to be low depth and becomes a standard algorithm. 
 
In this paper we propose a simple method to obtain low-depth algorithms. If we take the classical Monte-Carlo method as vertical and standard quantum algorithms as horizontal in a very literal sense, then interpolating the two simply amounts to moving horizontally and vertically at the same time. In other words, what we are going to do is to run the classical Monte-Carlo method on standard quantum algorithms themselves. See \cref{fig:main-idea} for an illustration of this idea.
\begin{figure}[ht]
 \centering
 \includegraphics[width=0.6\linewidth]{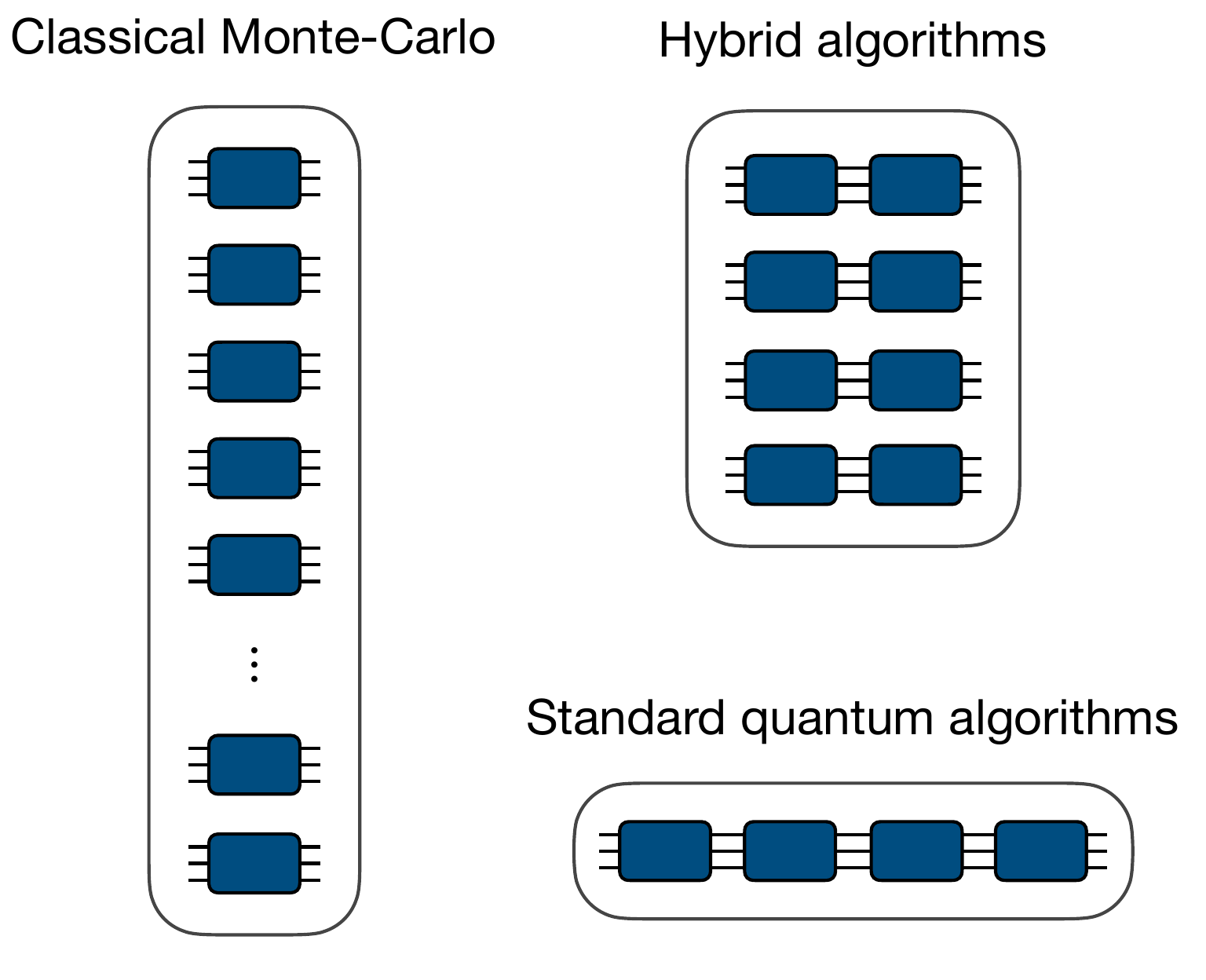}
 \caption{A hybrid construction of low-depth algorithms. It avoids the worst of both worlds: a lower depth compared to standard quantum amplitude estimation algorithm and a smaller query complexity compared to the classical Monte-Carlo method.}
 \label{fig:main-idea}
\end{figure} 

Not all quantum algorithms can be used as Monte-Carlo samples however. Due to the law of large numbers, the bias and the precision of the hybrid estimator are roughly equal. Since the bias is preserved by the Monte-Carlo method, it means that the constituent quantum algorithm, despite its low depth, must have a bias at least as small as the precision of a quantum algorithm with a larger depth. Put differently, the maximum depth of the constituent quantum algorithm must have a weak dependence on it's bias. Algorithms with this property are often referred to as \textit{unbiased}, a concise yet potentially misleading name. Several examples can be found in \cite{vanapeldoorn2022quantumtomographyusingstatepreparation,cornelissen_sublinear-time_2023}. The purpose of this paper is to show that

\textbf{Claim}: \textit{Any unbiased amplitude estimation algorithm can be turned into a low-depth algorithm.}

Our method works equally well for the task of phase estimation. In particular we present three low-depth amplitude/phase estimation algorithms corresponding to the three existing unbiased amplitude/phase estimation algorithms found in the literature. Our method is parallel and in some cases it can even inherit the unbiasedness of the constituent algorithm. 

Up until this point unbiased algorithms have been created to serve different purposes. The one in \cite{vanapeldoorn2022quantumtomographyusingstatepreparation} is used in state tomography where the idea is that unbiased estimations of the density matrix elements lead to a smaller $l_q$-norm error ($q\geq 2$) of the state.  The one in \cite{cornelissen_sublinear-time_2023} is used to compute a particular graph partition function. The partition function is written as a telescoping product and it was shown that unbiased estimations of the individual terms lead to a more accurate estimate of the product. Here we bring all unbiased algorithms under the same roof by showing that they can all be used to obtain high precision estimates from limited-depth hardware. Overall one can say that unbiased algorithms are not interesting by themselves, but in situations where we have to combine estimators in one way or another then having unbiased property might be what we need, and all we need. 

The paper is structured as follows. In \cref{sec:motivation}, we provide a more detailed context that motivated our work. In \cref{sec:unbiased_to_low_depth} and \cref{sec:low_depth_phase_estimation}, we provide the general protocol for building low-depth algorithms from unbiased ones, including both amplitude and phase estimation. In \cref{sec:examples}, we provide some concrete examples of this construction. Finally, \cref{sec:conclusion} concludes and discusses future directions.

\section{Motivation}
\label{sec:motivation}

\subsection{Hardware limitation and low-depth algorithms}
Imagine that we are faced with an estimation problem which requires an additive precision of $\epsilon_\circ$ of the amplitude $a_\circ \in [0,1]$ with success probability at least $1-\delta_\circ$. We are equipped with a quantum hardware that can handle at most $D_{\text{hw}}$ Grover operators sequentially. In terms of the algorithm, we have at our disposal an quantum amplitude estimation (QAE) algorithm $\text{QAE}(a,\epsilon,\delta)$ which can output an estimator with additive precision $\epsilon$ and success probability at least $1-\delta$ using a maximum of $D_{\text{QAE}}(a,\epsilon,\delta)$ consecutive Grover operators, which we refer to as the depth of the circuit. 
We take a user-centric point of view in which the algorithm is considered a black box and all we know about it is the maximum depth $D_\text{QAE}$ it uses. We try to answer the following question: \textit{can we achieve the desired precision and success probability using the provided hardware and algorithm while still maintaining some quantum advantage?}

Here, quantum advantage means that the total query complexity is superior than that of classical Monte Carlo for the same precision and success probability. In principle, we know that $D_\text{QAE}(a,\epsilon,\delta) = \tilde{O}(\log (1/\delta)/\epsilon)$, where $\tilde{O}(\cdot)$ suppresses polylog factors in $\epsilon^{-1}$. Even though we do not know $a$, assume that the dependence on $a$ is slow-varying so that we can have a rough estimate $D_\text{QAE}(a_\circ,\epsilon, \delta_\circ)\approx c/\epsilon$ with a certain constant $c$. 
However, the hardware limitation $D_{\rm hw}$ imposes an upper bound on the achievable precision, given by $\epsilon_{\rm hw} := c/D_{\rm hw}$, leading to $\epsilon \geq \epsilon_{\rm hw}$.

\begin{figure}[t]
    \centering
    \includegraphics[width=0.4\linewidth]{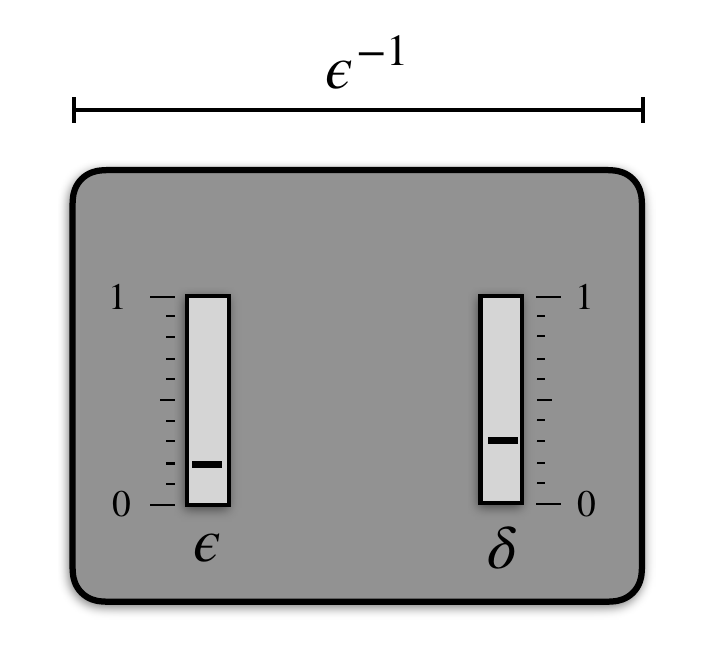}
    \caption{A standard black box QAE}
    \label{fig:black-box}
\end{figure}

For illustration, if $D_\text{hw}=10$ and $c=1$, then we cannot go below the precision of $0.1$. We assume the desired precision $\epsilon_\circ < \epsilon_\text{hw}$ in the following and we let $\beta\in (0,1]$ be a number such that
\begin{align}
    \epsilon_{\text{hw}} = \epsilon_\circ^{1-\beta}.
\end{align}
This number quantifies the limitation of our hardware: smaller $\beta$ corresponds to deeper hardware. Of course one can use the median trick to boost the success probability.
So in the definition of $\beta$, the failure probability $\delta$ is implicitly taken to be a constant smaller than $1/2$, say, $1/3$.
That is, $\epsilon_\text{hw}$ is the solution of the equation
\begin{align}
    D_\text{hw}= D_\text{QAE}(a_\circ, \epsilon_\text{hw},1/3).
\end{align}
However, at this stage we just want to convey the general idea instead of making precise statements.
So, $\epsilon_\text{hw}$ can simply be understood as the inverse circuit depth.
\begin{figure}[t]
    \centering
    \includegraphics[width=0.6\linewidth]{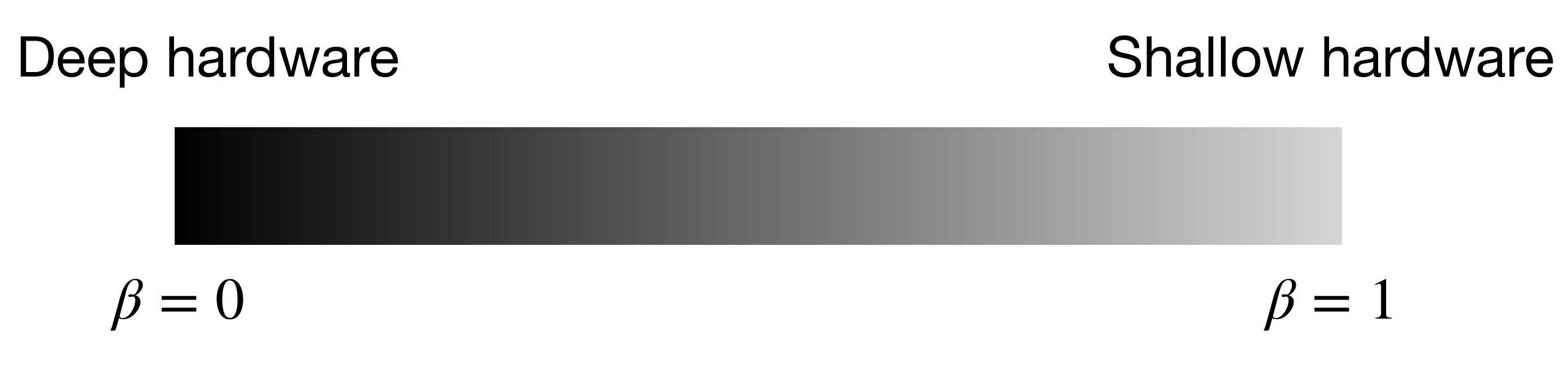}
    \caption{Hardware depth characterised by a single number $\beta$.}
    \label{fig:beta}
\end{figure}

This discussion motivates us to define low-depth algorithms, where for given values of $\epsilon$ and $\delta$ we have the ability to tune the maximum circuit depth through a parameter $\beta\in [0,1]$ in order to fit our hardware. According to \cite{burchard2019lowerboundsparallelquantum}, the maximum circuit depth $D$ and the total query complexity $N$ should satisfy the optimal condition $DN= O(\epsilon^{-2})$.
\begin{defn}
    Low-depth quantum amplitude estimation is an algorithm that can output an estimator with additive precision $\epsilon$ and success probability at least $1-\delta$ using a maximum circuit depth $D= \Tilde{O}(\epsilon^{-1+\beta})$ and a total query complexity $N= \Tilde{O}(\epsilon^{-1-\beta}\log(1/\delta))$ for any value of $\beta\in [0,1]$.
\end{defn}
\begin{figure}[ht]
     \centering
     \begin{subfigure}[t]{0.4\linewidth}
         \centering
         \includegraphics[width=0.9\linewidth]{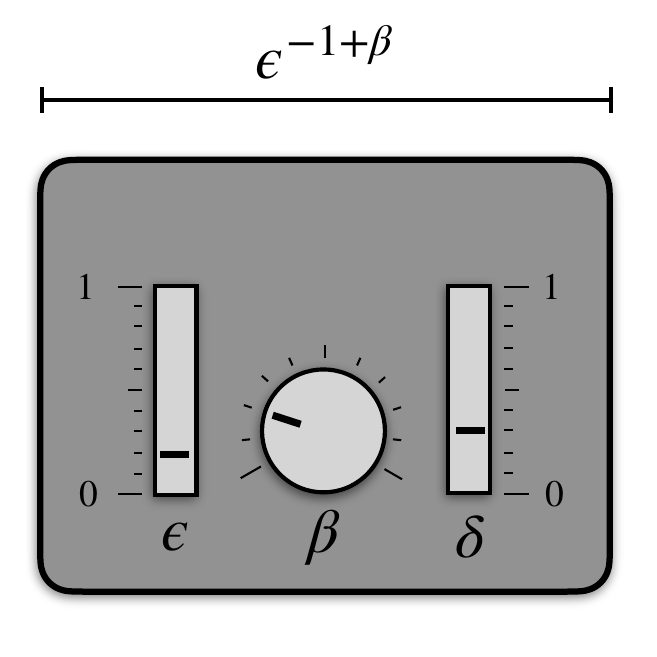}
         \caption{Low-depth QAE}
         \label{Low-depth QAE}
     \end{subfigure}
     ~ 
     \begin{subfigure}[t]{0.56\linewidth}
         \centering
         \raisebox{-1mm}{\includegraphics[width=0.82\linewidth]{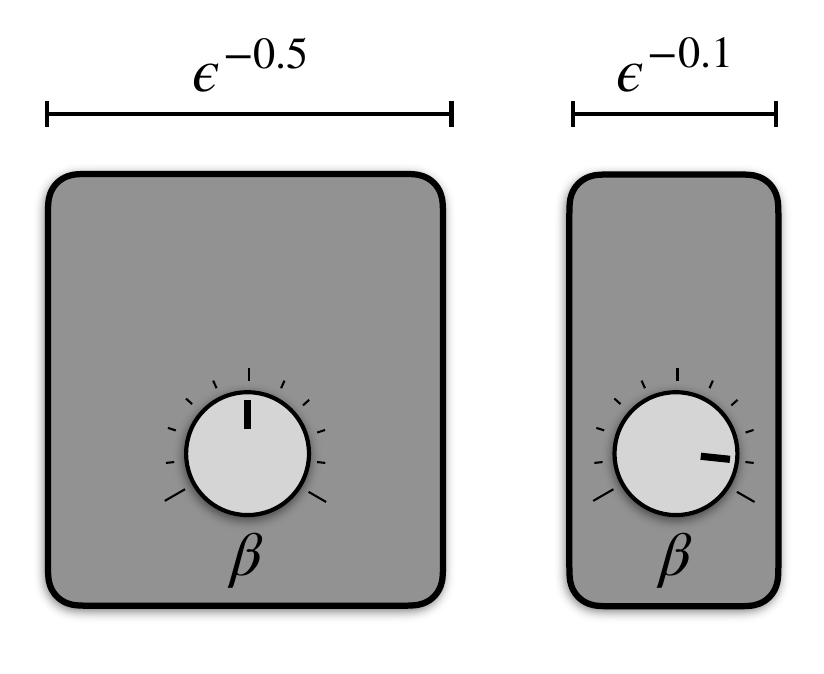}}
         \caption{Tuning $\beta$}
         \label{tuning}
     \end{subfigure}
     \caption{Low-depth QAE comes with an extra controlling knob that can tune the maximum depth of the algorithm without compromising $\epsilon$ and $\delta$, allowing it to fit into arbitrary hardware.}
\end{figure}
With the introduction of low-depth algorithms, we can recast our question into an equivalent form, which is the main focus of this work

\textbf{Question}: \textit{Can we use a black box QAE to form a low-depth QAE?}
\subsection{Aggregating the black boxes}
The simplest way to achieve a low-depth QAE is to run a standard black box $\text{QAE}(\epsilon_\circ^{1-\beta},\delta_\circ/T)$ independently for a total of $T$ times and try to combine the low-precision outputs to form a more accurate estimator. The failure probability has been divided by $T$ so that the final probability of success can be bounded from below by $1-\delta_\circ$ by the union bound. Of course we can also try to do this adaptively but doing so often requires us to ``open the black box". Moreover, independent runs allow for parallelization which has obvious computational advantage.
So, we focus on independent runs in this work.
For each $i\in \lbrace 1,2,...,T\rbrace$, $\text{QAE}(\epsilon_\circ^{1-\beta},\delta_\circ/T)$ outputs $\tilde{a}_i$ such that
\begin{align}
    P[|\tilde{a}_i-a|\leq \epsilon_\circ^{1-\beta}]\geq 1-\delta_\circ/T
\end{align}
In the event that $|\tilde{a}_i-a|\leq \epsilon_\circ^{1-\beta}$ for all $i\in \lbrace 1,2,...,T\rbrace$, which happens with probability at least $1-\delta_\circ$, we can try to construct a function $\hat{a}=F(\tilde{a}_1,\tilde{a}_2,...,\tilde{a}_T)$ and hope that $\hat{a}$ is the desired estimator, i.e., $|\hat{a}-a|\leq \epsilon_\circ$. 

This hope proves to be short-lived in view of the fact that what we have is a black box, and nothing forbids it to be a worst-case black box. 
As a simple example, consider what we call a monkey algorithm, defined to be one that outputs $a-\epsilon$ surely. Independent runs of this algorithm give the same estimator and nothing can be done to improve its accuracy. The monkey algorithm does what it is asked to do with minimal information since its output has zero entropy.

The monkey algorithm is a hypothetical concept that we use to convey the idea that if all we know about an algorithm is that it outputs some value in $[a-\epsilon,a+\epsilon]$ with probability at least $1-\delta$ (like the vanilla QAE algorithms), then that information alone is not enough for us to form high-precision estimators from low-precision ones. The algorithm might output near an extremal point of the interval, concealing valuable information around its center.  The existence of such idiosyncratic algorithms does not violate any principles.
Overall, it can be seen that \emph{vanilla QAE algorithms may not be sufficient to be combined to form a low-depth QAE}.

The class of vanilla QAE algorithms is too large for our question to have a positive answer and we need to impose some restrictions. The goal of this work is to show that

\textbf{Main result}: \textit{Only unbiased QAE can be used to form a low-depth QAE.}

We emphasize that the terminology ``unbiased" should be taken with care one since no quantum amplitude/phase estimation algorithm in the literature is truly unbiased. What is desired in this work is that the dependence of the maximum circuit depth and total query complexity on the bias is weak, say polylogarithmic in the bias. 
We will give two prototypes of unbiased QAE algorithms in \cref{sec:unbiased_to_low_depth}.
But before that, we first give an intuitive explanation why the unbiased property is required.

\subsection{Why unbiased?}
Simply put, the admissible class of QAE algorithms needs to be sufficiently restricted that even the monkey algorithm of that class can be used to form a low-depth QAE. Now for monkey algorithms the bias and the precision are equal so requiring the latter to be small automatically imposes the same condition on the former. A less \textit{ad-hoc} argument is the following. Let's say we have run an admissible algorithm a number of times and have obtained $\tilde{a}_1,...,\tilde{a}_T$ as low-precision estimators from which we construct $\hat{a}=F(\tilde{a}_1,...,\tilde{a}_T)$ which has high precision. $F$ can be any function but the simplest one is the sample mean $\bar{a} := (\tilde{a}_1+...+\tilde{a}_T)/T$. The law of large numbers says that for large $T$, the sample mean becomes very close to $a-B[\tilde{a}]$ (or $a+B[\tilde{a}]$) where $B[\tilde{a}]:= |\EE[\tilde{a}]-a|$ is the bias. Requiring the precision of $\bar{a}$ to be small again imposes the same condition on the bias $B[\tilde{a}]$.
\begin{figure}[ht]
    \centering
    \includegraphics[width=0.6\linewidth]{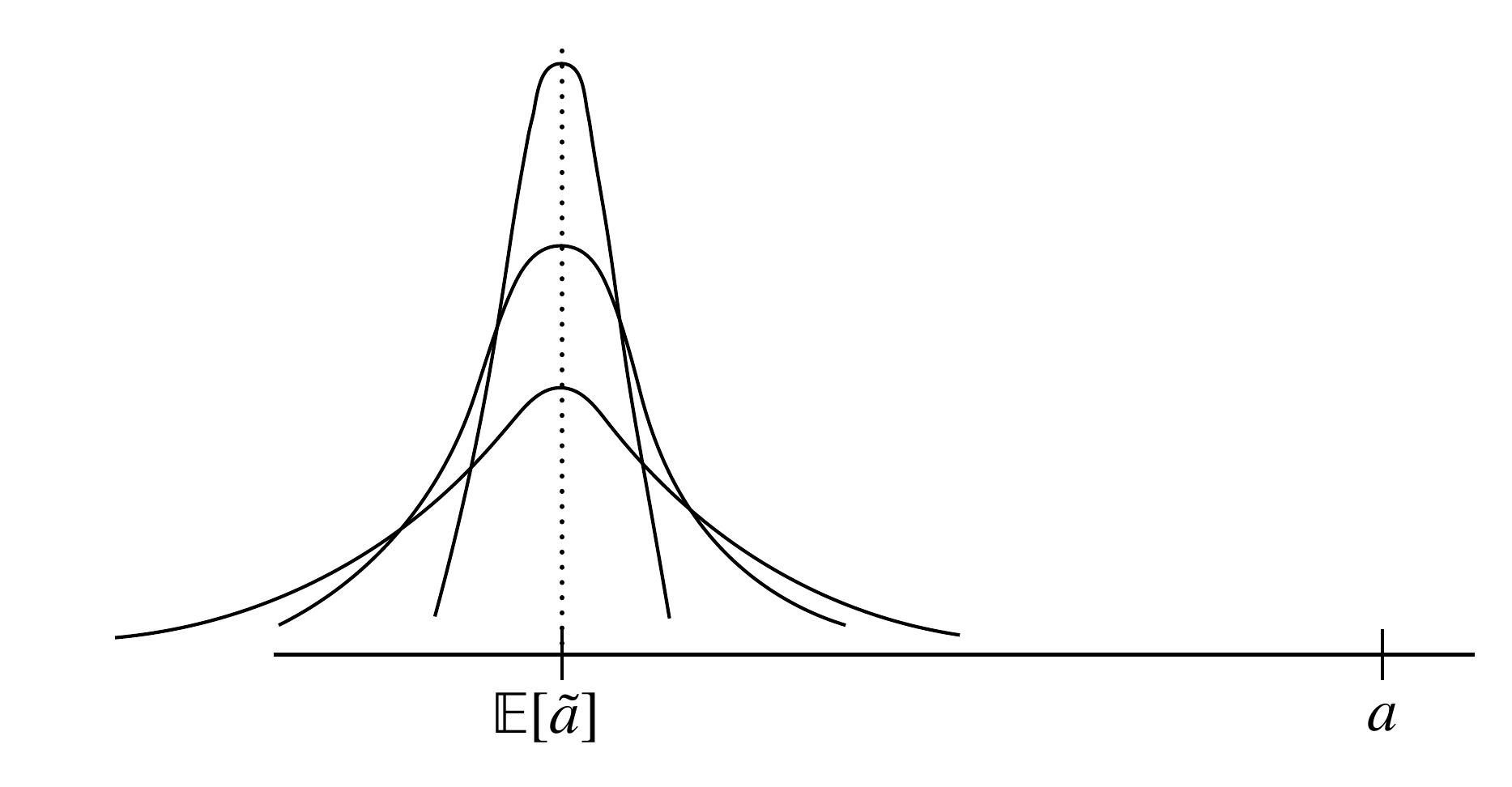}
    \caption{The convergence of the sample mean.}
    \label{fig:monki}
\end{figure}

Overall, our argument is that any admissible class of QAE algorithms might contain small variance algorithms (with the monkey as an extreme example), or the process of forming the low-depth algorithm itself induces a reduction in variance (as in the case of the sample mean). When the variance is small, the bias is required to be small as well otherwise the precision will be low. The mean estimator not only serves as an heuristic for the necessary part of our claim but also is the construction for sufficiency.

\section{From unbiased QAE to low-depth QAE}
\label{sec:unbiased_to_low_depth}

In this section, we show that any unbiased QAE can be used to form a low-depth QAE. We first provide two prototypes that conceptualize the available unbiased QAE algorithms in the literature. Some concrete examples will be given in \cref{sec:examples}. We denote the maximum circuit depth by $D$ and total query complexity by $N$ for the rest of the paper.
\begin{defn}[Unbiased QAE type I]
\label{prototype-I}

Given any $B,V>0$, an unbiased QAE of type I (UQAE-I) is an algorithm that can output an estimator $\tilde{a}$ of $a$ such that
\begin{align}
    |\EE[\tilde{a}]-a|\leq B \quad \text{and}\quad \Var[\tilde{a}]\leq V,
\end{align}
with maximum depth and total query complexity 
\begin{align}
    D_\text{UQAE-I}(B,V) = \tilde{O}(V^{-1/2})\quad \text{and}\quad N_\text{UQAE-I}(B,V) = \tilde{O}(V^{-1/2}),
\end{align}
where $\tilde{O}$ suppresses polylogarithmic terms in $V^{-1}$ and $B^{-1}$.
\end{defn}
\begin{defn}[Unbiased QAE type II] 
\label{prototype-II}
Given any $B,\epsilon,\delta>0$, an unbiased QAE of type II (UQAE-II) is an algorithm that can output an estimator $\tilde{a}$ of $a$ such that
\begin{align}
    |\EE[\tilde{a}]-a|\leq B \quad \text{and}\quad P[|\tilde{a}-a|\leq \epsilon]\geq 1-\delta,
\end{align}
with maximum depth and total complexity
\begin{align}
    D_\text{UQAE-II}(B,\epsilon,\delta) = \tilde{O}(\epsilon^{-1})\quad \text{and}\quad N_\text{UQAE-II}(B,\epsilon,\delta) = \tilde{O}(\epsilon^{-1}\log(\delta^{-1})),
\end{align}
where $\tilde{O}$ suppresses polylogarithmic terms in $\epsilon^{-1}$ and  $B^{-1}$.
\end{defn}

We show that for both types of UQAE, simply running the algorithm independently a number of times and 
taking the sample mean gives us a low-depth algorithm. The main idea is to set the bias to be of the same order as the desired accuracy and then use Chebyshev (for type I) or Hoeffding inequality (for type II) to bound the success probability.
\subsection{Type I: bias and variance}

\begin{problem} 
\label{problem-1}
For any $\beta\in [0,1]$, find $B,V,T$ such that running a black box $\text{UQAE-I}(B,V)$ as defined in \cref{prototype-I} independently $T$ times and averaging the outputs $\tilde{a}_i$ gives an estimator $\bar{a}:=\frac{1}{T} \sum_{i=1}^T \tilde{a}_i$ satisfying
\begin{align}
    P[|\bar{a}-a|\leq \epsilon_\circ]\geq 1-\delta_\circ
\end{align}
with circuit depth $D = \tilde{O}(\epsilon_\circ^{-1+\beta})$ and query complexity $N = \tilde{O}(\epsilon_\circ^{-1-\beta}\log(\delta_\circ^{-1}))$.
\end{problem}
By definition, solving this problem constitutes a protocol to build low-depth algorithms from black box UQAE-I.

\textbf{Analysis}: Since independent runs of the black box uses the same parameters $B$ and $V$, the maximum circuit depth of the aggregated algorithm is equal to that of each individual black box. Given that $D_\text{UQAE-I}(B,V) = \tilde{O}(V^{-1/2})$, requiring $D = \tilde{O}(\epsilon_\circ^{-1+\beta})$ means that we have to set $V= O(\epsilon_\circ^{2-2\beta})$. With the scaling of $V$ obtained, we can now deduce the scaling of $T$. The variance of the sample mean can be bounded as follows
\begin{align}
    \Var[\bar{a}] = \frac{1}{T^2}\sum_{i=1}^T\Var[\bar{a}_i]\leq \frac{V}{T}.
\end{align}
If we want the confidence interval of $\bar{a}$ to be of length $\epsilon_\circ$, then the variance of $\bar{a}$ needs to be of order $\epsilon_\circ^2$. As a consequence $T= 
 O(\epsilon_\circ^{-2\beta})$. Now that we have an adequately small variance, we just need to make sure that the bias is also under control. The following lemma makes precise this idea
 
\begin{lemma}[Bias and variance] 
\label{lemma:bias_and_variance}
Let $r\in (0,1)$ and $s$ such that $s<\frac{1}{2}(1-r)^2$. If an estimator $\hat{a}$ of $a$ satisfies
\begin{align}
    |\EE[\hat{a}]-a|\leq r\epsilon \quad \text{and}\quad \Var[\hat{a}]\leq s\epsilon^2.
\end{align}
Then 
\begin{align}
    P[|\hat{a}-a|\leq \epsilon]>\frac{1}{2}.
\end{align}
\end{lemma}

\begin{proof} According to the triangle inequality $|\hat{a}-a|\leq |\hat{a}-\EE[\hat{a}]|+|\EE[\hat{a}]-a|$, thus if  $|\hat{a}-\EE[\hat{a}]|+|\EE[\hat{a}]-a|\leq \epsilon$ then $|\hat{a}-a|\leq \epsilon$ i.e. the former event is a subset of the latter. It follows that 
\begin{align}
    P[|\hat{a}-a|\leq \epsilon] \geq  P[|\hat{a}-\EE[\hat{a}]|+|\EE[\hat{a}]-a|\leq \epsilon].
\end{align}
Using similar logic, since $|E[\hat{a}]-a|\leq r\epsilon$
\begin{align}
    P[|\hat{a}-\EE[\hat{a}]|+|\EE[\hat{a}]-a|\leq \epsilon]\geq  P[|\hat{a}-\EE[\hat{a}]|\leq (1-r)\epsilon].
\end{align}
Applying Chebyshev's inequality
\begin{align}
    P[|\hat{a}-\EE[\hat{a}]|\leq (1-r)\epsilon]\geq 1-\frac{\Var[\hat{a}]}{(1-r)^2\epsilon^2} \geq 1-\frac{s}{(1-r)^2}.
\end{align}
The last quantity is strictly greater than $1/2$ as per the assumption on $s$. 
\end{proof}

This lemma basically says that if the bias and the variance are respectively of order $\epsilon_\circ$ and $\epsilon_\circ^2$ (with a certain constraint between their prefactors), then a confidence interval of length $2\epsilon_\circ$ around the true value is guaranteed with probability greater than one half. We summarise the obtained parameters in the following table before proceeding to the proof.
\begin{table}[ht]
\centering
  \begin{tabular}{ | c | c | c  |}
    \hline
    \multicolumn{3}{|c|}{Choose constants $r,s$ with $r\in (0,1)$ and $s<(1-r)^2/2$}\\
    \hline
    Parameter & Value & Intuitive meaning of the value \\ \hline
    Bias $B$ & $r\epsilon_\circ$ & a fraction of the desired accuracy \\ \hline
    Variance $V$ & $s\epsilon_\circ^{2-2\beta}$ & square of inverse hardware depth \\ \hline
    Parallel runs $T$ & $\lceil\epsilon_\circ^{-2\beta}\rceil$ & hardware-limited variance divided by desired variance \\
    \hline
  \end{tabular}
  \caption{Parameter setting for transforming type-I UQAE to a low-depth QAE.}
  \label{protocol-1}
\end{table}
\begin{theorem}
    Given $\beta\in [0,1]$, $r \in (0, 1)$ and $s < (1 - r)^2/2$, set $B = r\epsilon_\circ$, $V = s\epsilon_{\circ}^{2-2\beta}$ and $T = \lceil \epsilon_\circ^{-2\beta} \rceil$.
    Then, running a black box $\text{UQAE-I}(B,V)$ as defined in \cref{prototype-I} independently $T$ times and averaging the outputs $\tilde{a}_i$ gives an estimator $\bar{a}:=\frac{1}{T} \sum_{i=1}^T \tilde{a}_i$ satisfying
    \begin{align}
        P[|\bar{a}-a|\leq \epsilon_\circ]\geq 1-\delta_\circ
    \end{align}
    with circuit depth $D = \tilde{O}(\epsilon_\circ^{-1+\beta})$ and query complexity $N = \tilde{O}(\epsilon_\circ^{-1-\beta}\log(\delta_\circ^{-1}))$.
\end{theorem}
\begin{proof}
    The bias of the mean is no greater than the mean of the individual biases. Indeed
    \begin{align}
        |\EE[\bar{a}]-a| = \frac{1}{T} \left|\sum_{i=1}^T (\EE[\tilde{a}_i] -a ) \right| \leq \frac{1}{T}\sum_{i=1}^T |\EE[\tilde{a}_i] -a|.
    \end{align}
    It follows that $|\EE[\bar{a}]-a|\leq r\epsilon_\circ = B$. Moreover
    \begin{align}
        \Var[\bar{a}]= \frac{1}{T^2}\sum_{i=1}^T \Var[\tilde{a}_i]\leq \frac{1}{T}s\epsilon_\circ^{2-2\beta}\leq s\epsilon_\circ^2.
    \end{align}
    \cref{lemma:bias_and_variance} thus guarantees that $P[|\bar{a}-a|\leq \epsilon_\circ]>\frac{1}{2}$. Regarding resources, the maximum circuit depth of the aggregated algorithm is equal to that of each individual run, which is equal to 
    \begin{align}
        D = D_{\text{UQAE-I}(B,V)} =  \Tilde{O}(V^{-1/2}) = \Tilde{O}(\epsilon_\circ^{-1+\beta}).
    \end{align}
    The total query complexity is equal to $T$ times the query complexity of each run
    \begin{align}
        N = TN_{\text{UQAE-I}(B,V)} = T \Tilde{O}(V^{-1/2}) = \Tilde{O}(\epsilon_\circ^{-1-\beta}).
    \end{align}
    In both expressions, $\Tilde{O}$ potentially suppresses polylogarithmic terms in $\epsilon_\circ^{-1}$. Using the median trick to boost the success probability to $1-\delta_\circ$ results in a total query complexity of $\Tilde{O}(\epsilon_\circ^{-1-\beta}\log(\delta_\circ^{-1}))$.
\end{proof}
\subsection{Type II: Bias and failure probability}
\begin{problem} 
\label{problem-2}
For any $\beta\in [0,1]$, find $B,\epsilon,\delta,T$ such that running a black box $\text{UQAE-II}(B,\epsilon,\delta)$ as defined in \cref{prototype-II} independently $T$ times and averaging the outputs $\tilde{a}_i$ gives an estimator $\bar{a}:=\frac{1}{T} \sum_{i=1}^T \tilde{a}_i$ such that
\begin{align}
    P[|\bar{a}-a|\leq \epsilon_\circ]\geq 1-\delta_\circ,
\end{align}
with resources $D = \tilde{O}(\epsilon_\circ ^{-1+\beta})$ and $N = \tilde{O}(\epsilon_\circ^{-1-\beta}\log(\delta_\circ^{-1}))$.
\end{problem}

\textbf{Analysis}: This problem is a bit more intricate compared to the previous one so we will start by making some assumptions and then remove them one by one.

\underline{\textit{Simplest model: zero bias and zero failure probability}}. Assume that the maximum circuit depth and the total query complexity have no dependence on $B$ and $\delta$ so we can set $B =\delta=0$. In that case we have an algorithm $\text{UQAE-II}(\epsilon)$ that outputs an estimate $\tilde{a}$ of $a$ such that $P[|\tilde{a}-a|\leq \epsilon]=1$ and $\EE[\tilde{a}]=a$ using a maximum depth $D=\Tilde{O}(\epsilon^{-1})$. Equating the maximum depth gives $\epsilon = \epsilon_\circ^{1-\beta}$ thus for all $i$ we have $\tilde{a}_i\in [a-\epsilon_\circ^{1-\beta},a+\epsilon_\circ^{1-\beta}]$  and $\EE[\tilde{a}_i]=a$. According to Hoeffding's inequality
\begin{align}
    P[|\bar{a}-a|\geq \epsilon_\circ]\leq 2\exp(-\frac{T\epsilon_\circ^2}{2\epsilon_\circ^{2-2\beta}})= 2\exp(-\frac{1}{2}T\epsilon_\circ^{2\beta}).
\end{align}
Requiring this probability to be smaller than or equal to $\delta_\circ$ amounts to setting 
\begin{align}
     T = 2\log(2\delta_\circ^{-1})\epsilon_\circ^{-2\beta}.
\end{align}
The number of parallel runs thus has the same scaling as in the case of type-I UQAE.

\underline{\textit{Intermediate model: zero bias but non-zero failure probability}}. Let us now consider an algorithm $\text{UQAE-II}(\epsilon,\delta)$ that can output an estimate $\tilde{a}$ of $a$ such that $P[|\tilde{a}-a|\leq \epsilon]\geq 1-\delta$ and $\EE[\tilde{a}]=a$ using a maximum depth $D= \Tilde{O}(\epsilon^{-1})$. Let us decompose $\tilde{a}$ into two components, depending on whether they belong to the confidence interval or not
\begin{align}
    \tilde{a} = \tilde{a} \mathbf{1}(\tilde{a}\in [a-\epsilon,a+\epsilon])+ \tilde{a} \mathbf{1}(\tilde{a}\notin [a-\epsilon,a+\epsilon]) := \tilde{a}^{(\text{good})} +\tilde{a}^{(\text{bad})}.
\end{align}
Our attention is on the good part, since if it is unbiased then we are back to the simplest model. However the presence of the bad part shifts the bias by a certain amount. Let us try to bound this shift. Let $C$ be a constant such that $\tilde{a}\in [-C,C]$. In principle $C$ is 1 but some algorithms might output a value larger than 1.
\begin{align}
    |\EE[\tilde{a}^{(\text{bad})}]|\leq \max(|\tilde{a}|) \EE[\mathbf{1}(\tilde{a}\notin [a-\epsilon,a+\epsilon])]\leq C\delta.
\end{align}
Thus $a+C\delta\geq \EE[\tilde{a}^{(\text{good})}]\geq a-C\delta$ i.e. the bias of the good part is bounded by $C\delta$. 

Let us now set $\epsilon = \epsilon_\circ^{1-\beta}$ while $\delta$ and $T$ are to be determined. Running $\text{UQAE-II}(\epsilon,\delta)$ independently $T$ times and obtain $\tilde{a}_1,\tilde{a}_2,...,\tilde{a}_T$. Denote by $G$ the event that $\tilde{a}_i = \tilde{a}_i^{(\text{good})}$ for all $i$. By the union bound $P(\bar{G})\leq T\delta$. We can thus bound the final failure probability as follows
\begin{align}
    P[|\bar{a}-a|\geq \epsilon_\circ] &= P[(|\bar{a}-a|\geq \epsilon_\circ) \cap G]+P[(|\bar{a}-a|\geq \epsilon_\circ) \cap \overline{G}]\\
    & \leq P[|\bar{a}^{(\text{good})}-a|\geq \epsilon_\circ]+T\delta,
\end{align}
where we have denoted $\bar{a}^{(\text{good})}:= (\tilde{a}_1^{(\text{good})}+...+\tilde{a}_T^{(\text{good})})/T$. If we want the final failure probability to be upper bounded by $\delta_\circ$, we can choose $T$ and $\delta$ such that $T\delta\leq \delta_\circ/2$ and then try to make $P[|\bar{a}^{(\text{good})}-a|\geq \epsilon_\circ]$ less than or equal to $\delta_\circ/2$. By the triangle inequality $|\bar{a}^{(\text{good})} - a|\leq |\bar{a}^{(\text{good})}-E[\bar{a}^{(\text{good})}]|+|E[\bar{a}^{\text{good}}]-a|$. The second term is the bias of the mean of good estimators, which cannot exceed the mean of the individual biases, which themselves are upper bounded by $C\delta$ as shown above. Thus
\begin{align}
    P[|\bar{a}^{(\text{good})}-a| \geq \epsilon_\circ] &\leq P[|\bar{a}^{(\text{good})}-\EE[\bar{a}^{(\text{good})}]|+|\EE[\bar{a}^{(\text{good})}]-a|\geq \epsilon_\circ]\\
    &\leq P[|\bar{a}^{(\text{good})}-\EE[\bar{a}^{(\text{good})}]|\geq \epsilon_\circ -C\delta]\leq 2 \exp(-\frac{T(\epsilon_\circ-C\delta)^2}{2\epsilon_\circ^{2-2\beta}}),
\end{align}
where we have applied Hoeffding's inequality in the last step. If we take $\delta \leq  s\epsilon_\circ/C$ for some constant $s<1$ and $T = \frac{2\log(4\delta_\circ^{-1})\epsilon_\circ^{-2\beta}}{(1-s)^2}$ then it's guaranteed that $P[|\bar{a}^{(\text{good})}-a| \geq \epsilon_\circ]\leq \delta_\circ/2$ as we need. Combine the two conditions on $\delta$ leads to the choice $\delta = \min(\frac{\delta_\circ}{2T}, \frac{s\epsilon_\circ}{C})$.

\underline{\textit{Realistic model: non-zero bias and non-zero failure probability}}. Reintroduce the bias simply shifts the mean of the good part by a corresponding amount, as long as this shift is of order $\epsilon_\circ$ the analysis is not much different from the previous case, see the proof for more details.

\begin{table}[ht]
\centering
  \begin{tabular}{ | c | c | c  |}
    \hline
    \multicolumn{3}{|c|}{Choose constants $r,s\in(0,1)$ such that $r+s<1$}\\
    \hline
    Parameter & Value & Intuitive meaning of the value \\ \hline
    Bias $B$ & $r\epsilon_\circ$ & a fraction of the desired precision \\ \hline
    Precision $\epsilon$ & $\epsilon_\circ^{1-\beta}$ & hardware limited precision \\ \hline
    Failure probability $\delta$ & $\min(\frac{\delta_\circ}{2T},\frac{s\epsilon_\circ}{C})$ & \begin{tabular}{@{}c@{}}some power of the desired precision that keeps \\ both the bias and the failure probability small\end{tabular} \\
    \hline
    Parallel runs $T$ & $\frac{2\log(4\delta_\circ^{-1})\epsilon_\circ^{-2\beta}}{(1-r-s)^2}$ & \begin{tabular}{@{}c@{}}square of the limited precision divided \\ by square of the desired precision\end{tabular} \\
    \hline
  \end{tabular}
  \caption{Parameter setting for transforming type-II UQAE to a low-depth QAE. $C$ is the maximum absolute value of the black box estimator.}
  \label{protocol-2}
\end{table}
\begin{theorem}
    Given $\beta\in [0,1]$ and $r, s \in (0, 1)$ such that $r + s < 1$, set $B = r\epsilon_\circ$, $\epsilon = \epsilon_{\circ}^{1 - \beta}$, $\delta = \min(\frac{\delta_\circ}{2T}, \frac{s\epsilon_\circ}{C})$ and $T = \frac{2\log(4\delta_\circ^{-1})\epsilon_\circ^{-2\beta}}{(1-r-s)^2}$.
    Then, running a black box $\text{UQAE-II}(B,\epsilon,\delta)$ as defined in \cref{prototype-II} independently $T$ times and averaging the outputs $\tilde{a}_i$ gives an estimator $\bar{a}:=\frac{1}{T} \sum_{i=1}^T \tilde{a}_i$ such that
    \begin{align}
        P[|\bar{a}-a|\leq \epsilon_\circ]\geq 1-\delta_\circ,
    \end{align}
    with resources $D = \tilde{O}(\epsilon_\circ ^{-1+\beta})$ and $N = \tilde{O}(\epsilon_\circ^{-1-\beta}\log(\delta_\circ^{-1}))$.
\end{theorem}
\begin{proof}
    Let $G$ be the event that $\tilde{a}_i = \tilde{a}_i^{(\text{good})} $ for all $i=1,...,T$. As a result of the union bound $P[\overline{G}]\leq T\delta \leq  \delta_\circ/2$, thus
    \begin{align}
        P[|\bar{a}-a|\geq \epsilon_\circ] &= P[(|\bar{a}-a|\geq \epsilon_\circ) \cap E]+P[(|\bar{a}-a|\geq \epsilon_\circ) \cap \overline{E}] \\
        & \leq P[|\bar{a}^{(\text{good})}-a|\geq \epsilon_\circ]+ \delta_\circ/2.
    \end{align}
    By the triangle inequality
    \begin{align}
        |\bar{a}^{(\text{good})}-a| \leq \underbrace{|\bar{a}^{(\text{good})}- \EE[\bar{a}^{(\text{good})}]|}_{\text{deviation from the good expectation}}+\underbrace{|\EE[\bar{a}^{(\text{good})}]-\EE[\tilde{a}]|}_{\text{shift of the good expectation} }+\underbrace{|\EE[\tilde{a}]-a|}_{\text{bias}}.
    \end{align}
The first term is the deviation of the sample mean of the good parts from its expectation, which will be bounded using Hoeffding's inequality. The second term is the shift of the expectation of the good part from that of the original estimator, which has been shown to be upper bounded by $C\delta$, which itself cannot exceed $s\epsilon_\circ$ by construction. The last term is simply the bias, which is chosen to be no greater than $r\epsilon_\circ$. It follows that
    \begin{align}
        P[|\bar{a}^{(\text{good})}-a|\geq \epsilon_\circ] &\leq P[|\bar{a}^{(\text{good})}- \EE[\bar{a}^{(\text{good})}]|\geq \epsilon_\circ(1-r-s)]\\
        & \leq 2\exp(-\frac{T\epsilon_\circ^2(1-r-s)^2}{2\epsilon^{2-2\beta}})\leq \delta_\circ/2.
    \end{align}
    Putting everything together we have
    \begin{align}
        P[|\bar{a}-a|\geq \epsilon_\circ]\leq \delta_\circ.
    \end{align}
    The resource analysis is similar to that of type-I UQAE.
\end{proof}

\section{Low-depth phase estimation}
\label{sec:low_depth_phase_estimation}

Our method can be applied to phase estimation with slight modifications. We distinguish two types of phase estimation algorithms, not in terms of statistics such as variance or failure probability but in terms of arithmetic: the usual arithmetic of $\mathbb{R}$ or the circular arithmetic of $\mathbb{R}/2\pi$. In the former case, all previous constructions remain valid, see \cref{thm:unbiased_phase_est} and proposition \ref{protocol-2} below for an example. The latter case however comes with some subtleties.  First, note that if we want to formulate an unbiased phase estimation algorithm using circular arithmetic then we must define the difference between two angles since the bias is the expectation of the difference between the estimated angle and the true angle. Given two angles $\theta$ and $\phi$, we define $\theta\ominus\phi$  to be the unique number $r\in [-\pi,\pi)$ such that $\theta-\phi=r \mod{2\pi}$.  It is the algebraic value of the angle between $\theta$ and $\phi$. Its absolute value is given by the length of the smaller arc and its sign depends on the relative position between $\theta$ and $\phi$: if we go along the this arc from $\theta$ to $\phi$ in a clockwise direction ($\theta$ is on the left of $\phi$) then the sign is positive, otherwise it is negative. Consider then the following prototype of an unbiased phase estimation algorithm formulated in terms of the circular difference.
\begin{defn}[Unbiased QPE type II] Given any $B,\epsilon,\delta>0$, an unbiased QPE of type II (UQPE-II) is an algorithm that can output an estimator $\tilde{\theta}$ of $\theta$ such that
\begin{align}
    |\EE[\tilde{\theta}\ominus\theta]|\leq B \quad \text{and}\quad P[|\tilde{\theta}\ominus\theta|\leq \epsilon]\geq 1-\delta
\end{align}
with resources
\begin{align}
    D_\text{UQAE-II}(B,\epsilon,\delta) = \tilde{O}(\epsilon^{-1})\quad \text{and}\quad N_\text{UQAE-II}(B,\epsilon,\delta) = \tilde{O}(\epsilon^{-1}\log(\delta^{-1}))
\end{align}
where $\tilde{O}$ suppresses polylogarithmic terms in $\epsilon^{-1}$ and $B^{-1}$.
\end{defn}
We cannot immediately replace the circular difference by the standard difference because the former is always less than or equal to the latter in absolute value. The example of $\theta=0$ and $\tilde{\theta}=2\pi-\epsilon$ is a clear illustration of this point.  If we insist on working with circular arithmetic then we run into another trouble: a continuous notion of the mean cannot be defined for angles.

\begin{figure}[ht]
     \centering
     \begin{subfigure}[t]{0.45\linewidth}
         \centering
         \includegraphics[width=0.9\linewidth]{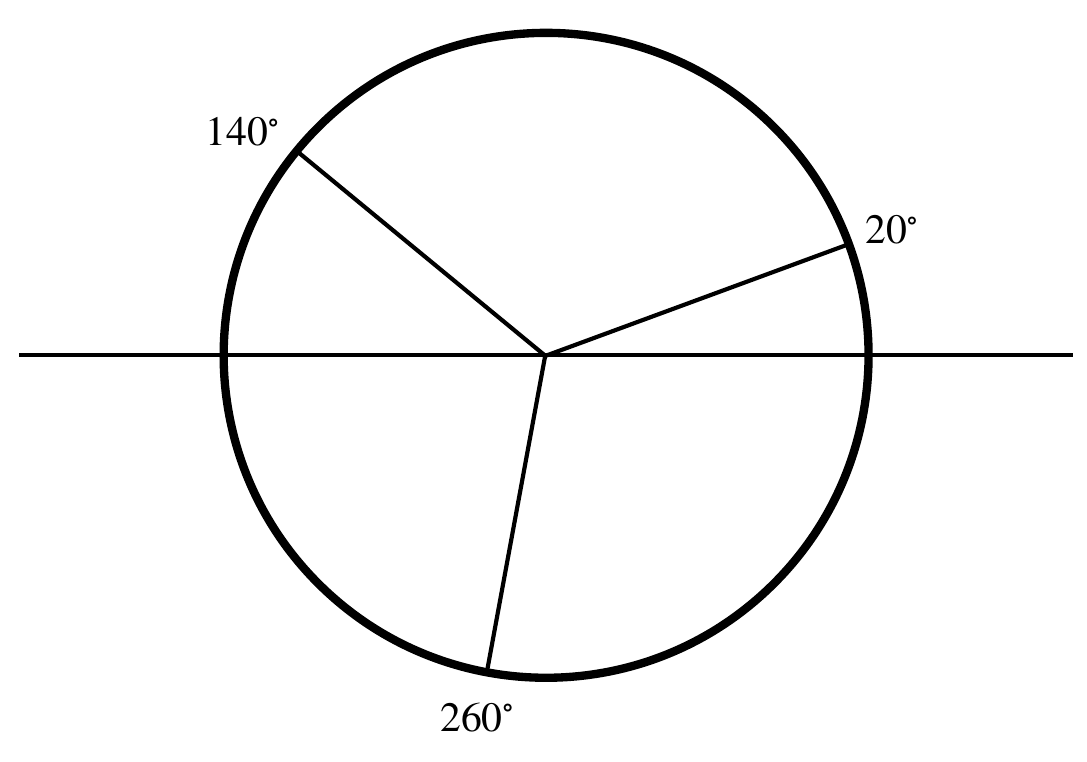}
         \caption{ Due to symmetry, no sensible notion of the mean can be defined for these angles.}
         \label{circular_1}
     \end{subfigure}
     \hfill 
     \begin{subfigure}[t]{0.45\linewidth}
         \centering
         \raisebox{1.1mm}{\includegraphics[width=0.9\linewidth]{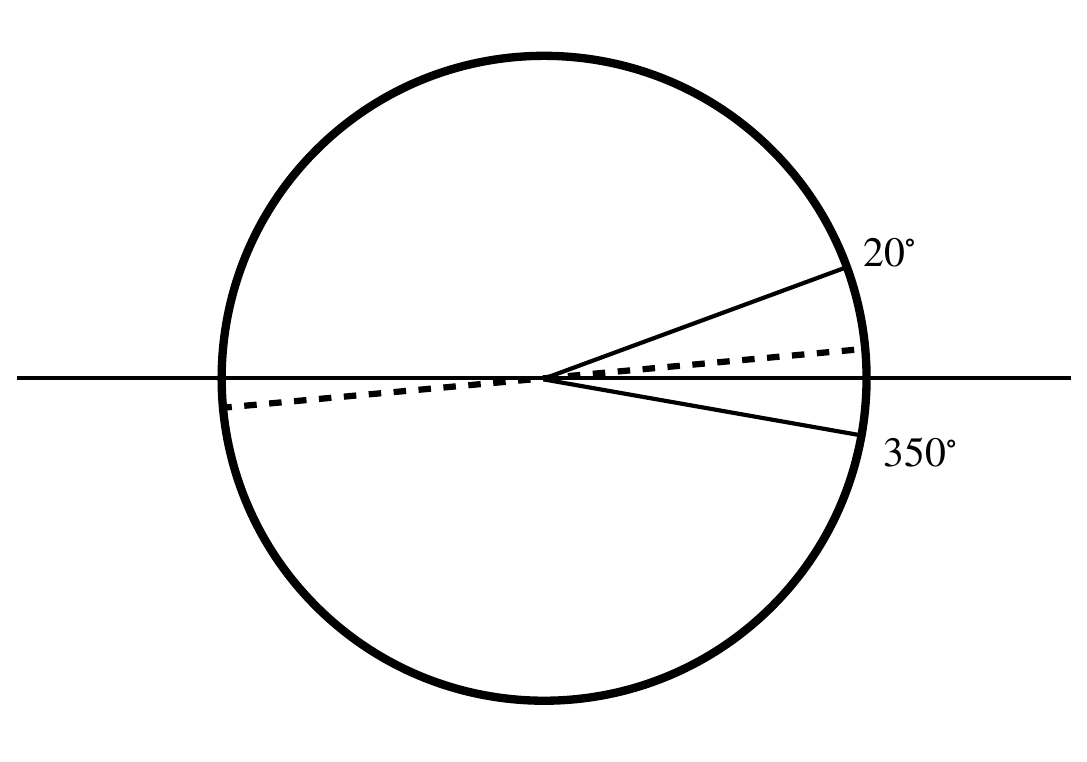}}
         \caption{The usual arithmetic mean gives $185^\circ$, which is not what we want in this case.}
         \label{circular_2}
     \end{subfigure}
     \caption{Some issues with defining a circular mean.}
     \label{circular-mean}
\end{figure}

\cref{circular-mean} shows two problems with defining a circular mean. First, when the angles spread out over the circle, it might not make sense to talk about the mean anymore. Second, even when the angles are concentrated inside an arc, the usual arithmetic mean might give us a value outside of that arc. We circumvent these issues by first obtaining a low-precision but high-confident estimate of $\theta$, building an arc that contains both $\theta$ and subsequent estimates with a high probability. We then construct an isometric and bijective mapping between angles inside this arc and the interval $[0,1]$. With this mapping, we can replace the circular difference by the usual difference and apply the previous technique to obtain a high-precision estimate. Finally, we apply the inverse mapping to convert the obtained value back to an angle. The failure probability of the entire process is upper bounded by $\delta_\circ$. We refer to the appendix for the detailed algorithm.

Can we formulate an UQPE in terms of bias and variance? The difficulty here lies in defining the circular variance itself and we leave this for future work.
\section{Examples}
\label{sec:examples}

In this section, we present several UQAE algorithms and show how to obtain a low-depth QAE from them. Previously we assume that parameters like bias and variance or bias and failure probability can be set independently. This is however not the case for most algorithms in the literature and some light manipulations are usually  needed in order to achieve the correct parameter setting.
\subsection{Type I UQAE}
\begin{theorem}[Unbiased amplitude estimation~\cite{cornelissen_sublinear-time_2023}]
\label{thm:unbiased_amp_est}
Given $K \geq 4$ and $B \in (0, 1)$, there exists an unbiased amplitude estimation algorithm that outputs an estimate $\tilde{a} \in [0,1]$ such that 
\begin{align}
    |\EE[\tilde{a}] - a| \leq B \quad \text{and}\quad \Var[\tilde{a}] \leq \frac{91a}{K^2} + B .
\end{align}
using $\order{K\log\log(K) \log(K/B)}$ total query complexity.
\end{theorem}
According to the established protocol \ref{protocol-1}, we need to find some constants $r\in (0,1)$ and $s\in (0,1/2)$ such that $s<\frac{1}{2}(1-r)^2$ and then set 
\begin{align}
    B \leq  r\epsilon_\circ \quad \text{and}\quad \frac{91a}{K^2} + B \leq s\epsilon_\circ^{2-2\beta}.
\end{align}
Apparently it might not be possible to set $B= r\epsilon_\circ$ especially if $\beta<1/2$ in which case $s\epsilon_\circ^{2-2\beta}$ can be smaller than $r\epsilon_\circ$. The issue here is that the bound for the variance is always greater than that of the bias so we cannot choose them recklessly. Lucky for us this is exactly where the power of being unbiased truly shines. Looking at the total query complexity we see that the dominant factor is indeed the inverse square root of the variance whereas the bias only has a logarithmic impact. What this means is that we cannot compromise on the variance but we can set the bias to any power of $\epsilon_\circ$ that we wish. In particular, we can take $B = \min (r\epsilon_\circ, s' \epsilon_\circ^{2-2\beta})$ with $0<s'<s$ and $K$ such that $\frac{91a}{K^2}\leq (s-s')\epsilon_\circ^{2-2\beta}$. The choice of $s=\frac{100}{225}$, $r=0.05$, $s'=9/225$ and $K=15\epsilon_\circ^{-1+\beta}$ satisfies all the requirements. Indeed $s<\frac{1}{2}(1-r)^2$ and 
\begin{align}
    \frac{91a}{K^2}\leq \frac{91}{K^2} = (s-s')\epsilon_\circ^{2-2\beta}.
\end{align}
To summarize, if we set $B = \min\big(\frac{9}{225}\epsilon_\circ^{2-2\beta},0.05\epsilon_\circ\big)$ and $K = 15 \epsilon_\circ^{-1+\beta}$ then we get the following bounds on bias and variance
\begin{align}
    |\EE[\tilde{a}] - a| \leq 0.05\epsilon_\circ,\quad \Var[\tilde{a}]\leq \frac{4}{9}\epsilon_\circ^{2-2\beta}
\end{align}
which conform to the protocol \ref{protocol-1}.
\begin{prop}[Low-depth amplitude estimation algorithm from \cref{thm:unbiased_amp_est}]
Let $T = \lceil \epsilon_\circ^{-2\beta}\rceil$, run the algorithm in \cref{thm:unbiased_amp_est} independently $T$ times with the following parameters
\begin{align}
     K = 15 \epsilon_\circ^{-1+\beta}, B = \min\big(\frac{9}{225}\epsilon_\circ^{2-2\beta},0.05\epsilon_\circ\big).
\end{align}
The sample mean $\bar{a}$ of the outputs satisfies $P[|\bar{a}-a|\leq \epsilon_\circ]>1/2$. The resources of obtaining $\bar{a}$ are
\begin{align}
    D = O(\epsilon_\circ^{-1+\beta}\log \log ( \epsilon_\circ^{-1})\log(\epsilon_\circ^{-1})),\quad N = O(\epsilon_\circ^{-1+\beta}\log \log ( \epsilon_\circ^{-1})\log(\epsilon_\circ^{-1}))
\end{align}
\end{prop}
Note that here we actually used a weaker bound for the variance of the unbiased amplitude estimation algorithm in \cref{thm:unbiased_amp_est}, namely $ \Var[\tilde{a}]\leq \frac{91}{K^2}+B$.
If we want to exploit the full power of the original bound, it comes with a cost of having a rough estimate of $a$ in the parameters $K$ and $B$, is it worth it? The authors of \cite{cornelissen_sublinear-time_2023} also provided an unbiased phase estimation algorithm 
\begin{theorem}[Unbiased phase estimation~\cite{cornelissen_sublinear-time_2023}]
\label{thm:unbiased_phase_est}
Let $U$ be a unitary operator with an eigenvector $|\psi\rangle$ such that $U|\psi\rangle = e^{2\pi i \theta}|\psi\rangle$ where $\theta\in [0,1/2]$. Given $K\geq 1$ and $B\in (0,1)$ there exists an unbiased phase estimation algorithm  which outputs an estimate $\tilde{\theta}\in [-1,1]$ such that
\begin{align}
    |\EE[\tilde{\theta}]-\theta|\leq B \quad \text{and}\quad \Var[\tilde{\theta}]\leq \frac{1}{K^2}+B
\end{align}
using $O(K\log(K/B))$ total query complexity.
\end{theorem}
With similar logic we obtain a low depth phase estimation algorithm
\begin{prop}[Low-depth phase estimation algorithm from \cref{thm:unbiased_phase_est}]
Let $T = \lceil \epsilon^{-2\beta}\rceil$, run the algorithm in \cref{thm:unbiased_phase_est} independently $T$ times with the following parameters
\begin{align}
     K = \sqrt{3} \epsilon_\circ^{-1+\beta}, B = \min\big(\frac{1}{9}\epsilon_\circ^{2-2\beta},0.05\epsilon_\circ\big).
\end{align}
The sample mean $\bar{a}$ of the outputs satisfies $P[|\bar{a}-a|\leq \epsilon_\circ]>1/2$. The resources of obtaining $\bar{a}$ are
\begin{align}
    D = O(\epsilon_\circ^{-1+\beta}\log(\epsilon_\circ^{-1})),\quad N = O(\epsilon_\circ^{-1+\beta}\log(\epsilon_\circ^{-1})).
\end{align}
\end{prop}
\subsection{Type-II UQAE}
Before showing an example of type-II unbiased algorithms, let us take a quick look at the following result, which was obtained using the technique of  semi-Pellian polynomials.
\begin{theorem}[Unbiased amplitude estimation \cite{Rall_2023}] For every $\epsilon, \delta, \eta > 0$ there exists an amplitude estimation algorithm satisfying
\begin{align}
    |\EE[\tilde{a}]-a|\leq \epsilon\eta+\delta,\quad\text{and}\quad P[|\tilde{a}-a|\geq \epsilon]\leq \delta
\end{align}
with maximum depth $D\in O(\epsilon^{-1}(\log(\delta^{-1})+\eta^{-1}))$.
\end{theorem}
Even though this algorithm has been claimed by its authors to be unbiased, we argue that the dependence of the maximum depth on the bias is too strong for it to be of any practical value. Indeed, the maximum depth here scales like inverse of the bias, but with such depth, a vanilla QAE can already output an estimator with a precision that is equal to that bias itself. The whole point of having an unbiased estimator is that one can tune the bias to be much smaller than the precision without drastically changing the maximum depth and total query complexity. Here for a desired precision $\epsilon$, if we take $\delta$ to be of order $O(\epsilon)$ and $\eta$ to be $O(1)$, then the bias is of the order $O(\epsilon)$ as well and that doesn't add much value to the estimator. On the other hand if we want the bias to be small, say $O(\epsilon^2)$ then $\delta$ must be taken to be $O(\epsilon^2)$ and $\eta$ to be $O(\epsilon)$, but then the maximum depth is $D= O(\epsilon^{-2})$. With such depth, even a vanilla QAE can output an estimator with precision $O(\epsilon^2)$, which is objectively speaking better than an estimator with precision $O(\epsilon)$ and bias $O(\epsilon^2)$. We take this opportunity to mention that up until this point there is no established standard in the literature as to what should be considered an unbiased amplitude/phase estimation algorithm. In this regard, our prototypes \ref{prototype-I} and \ref{prototype-II} are designed to serve as minimal baselines.
\begin{theorem}[Unbiased phase estimation \cite{vanapeldoorn2022quantumtomographyusingstatepreparation}] Given integers $n,m,M$ such that $n\geq \log_2(\pi m)$ then there exists an algorithm that outputs $\tilde{\theta}\in [0,2\pi]$  such that
\begin{align}
    |\EE[\tilde{\theta}\ominus\theta)]| &\leq 32\pi(m+1)2^{-n}\\
    P[|\tilde{\theta}\ominus\theta|\geq \frac{10}{M}(1+2^{-n})] &\leq 4\pi(m+1)2^{-n}+2e^{-m/4}
\end{align}
Using maximum depth of $M/2$ and total query complexity of $(2m+1)M$.
\end{theorem}

In this example, the three parameters $n,m$ and $M$ are related to the bias, the precision and the failure probability in a rather complicated way. According to \cref{uqpe-ii-setting}, they should satisfy the following conditions
\begin{align}
    \begin{cases}
    \frac{10}{M}(1+2^{-n}) = \epsilon_\circ^{1-\beta}\quad [\text{accuracy } \epsilon]\\
    32\pi (m+1)2^{-n} \leq r\epsilon_\circ \quad [\text{bias } B]\\
    4\pi(m+1)2^{-n}+2e^{-m/4} \leq  \min\bigg(\frac{\delta_\circ(1-r-s)^2}{4\log (4\delta_\circ^{-1})}\epsilon_\circ^{2\beta},\frac{s\epsilon_\circ}{4\pi}\bigg) \quad [\text{failure probability } \delta]
    \end{cases}
    \label{system-conditions}
\end{align}
where $r$ and $s$ are some positive constants such that $r+s<1$. Similar to the algorithm of \cref{thm:unbiased_amp_est} there is a constraint between the bias and the failure probability ($B\leq 8\delta$) which makes setting all inequalities to equality impossible for $\beta\geq 1/2$. To circumvent this difficulty we are going to make two simplifications. First, we note that
\begin{align}
    \epsilon_\circ^2\min\bigg(\frac{\delta_\circ(1-r-s)^2}{4\log(4\delta_\circ^{-1})},\frac{s}{4\pi}\bigg)\leq \min\bigg(\frac{\delta_\circ(1-r-s)^2}{4\log(4\delta_\circ^{-1}) }\epsilon_\circ^{2\beta},\frac{s\epsilon_\circ}{4\pi}\bigg)
\end{align}
This amounts to setting $\delta = O(\epsilon_\circ^2)$ instead of $O(\epsilon_\circ^{\max(1,2\beta)})$ which doesn't make any real difference since the total query complexity involves $\log(\delta^{-1})$ only. Then, the above mentioned constraint becomes easier to manage if we set $B=r\epsilon_\circ^2$, which automatically satisfies the bias condition. Overall, we can replace the above system of conditions by a weaker system of equations 
\begin{align}
    \begin{cases}
    \frac{10}{M}(1+2^{-n}) = \epsilon_\circ^{1-\beta}\quad [\text{accuracy } \epsilon]\\
    32\pi (m+1)2^{-n} = r\epsilon_\circ^2 \quad [\text{bias } B]\\
    4\pi(m+1)2^{-n}+2e^{-m/4} = \epsilon_\circ^2\min\bigg(\frac{\delta_\circ(1-r-s)^2}{4\log(4\delta_\circ^{-1})},\frac{s}{4\pi}\bigg) \quad [\text{failure probability } \delta]
    \end{cases}
    \label{system-equations}
\end{align}
Solutions of \cref{system-equations} are guaranteed to satisfy \cref{system-conditions}. The constraint between the bias and the failure probability now becomes a constraint between the constants $r$ and $s$
\begin{align}
    r < 8 \min\bigg(\frac{\delta_\circ(1-r-s)^2}{4\log (4\delta_\circ^{-1})},\frac{s}{4\pi}\bigg).
\end{align}
We can choose $r=\frac{\delta_\circ}{4\log(4\delta_\circ^{-1})}$ and $s=1/2-r$. The last two equations of \cref{system-equations} then become
\begin{align}
    \begin{cases}
        \pi(m+1)2^{-n} = \frac{\epsilon_\circ^2\delta_\circ}{128\log(4\delta_\circ^{-1})}\\
        e^{-m/4} = \frac{\epsilon_\circ^2\delta_\circ}{64\log(4\delta_\circ^{-1})}
    \end{cases}
\end{align}
\begin{prop}[Low-depth algorithm from \cref{thm:unbiased_phase_est}]
Setting the parameters 
\begin{align}
    m = -4 \log\bigg(\frac{\epsilon_\circ^2\delta_\circ}{64\log(4\delta_\circ^{-1})}\bigg),\quad  n =\log_2\big( 128\pi(m+1)\log(4\delta_\circ^{-1})\epsilon_\circ^{-2}\big),\quad M = \frac{\epsilon_\circ^{1-\beta}}{10(1+2^{-n})}
\end{align}
provides an estimate $\bar{\theta}$ such that
\begin{align}
    P[|\bar{\theta}\ominus\theta|\leq \epsilon_\circ]\geq 1-\delta_\circ
\end{align}
with resources
\begin{align}
    D = O(\epsilon_\circ^{-1+\beta}),\quad N = O\bigg(\epsilon_\circ^{-1+\beta}\log(\epsilon_\circ^{-1})\log(\delta_\circ^{-1})\log(\log(\delta_\circ^{-1}))\bigg).
\end{align}
\end{prop}

\section{Conclusion}
\label{sec:conclusion}

This work presents a universal protocol to transform any unbiased amplitude/phase estimation algorithm into a corresponding low-depth algorithm. At the moment of writing, all unbiased algorithms found in the literature unfortunately rely on the quantum Fourier transform (QFT) subroutine, making them unfit for near-term devices. The only attempt so far to construct a QFT-free unbiased amplitude estimation algorithm is that of \cite{Rall_2023} but the dependence of the maximum depth on the bias turns out to be too strong. Given their potential, we strongly believe that finding a QFT-free unbiased algorithm should be the focus of near-term amplitude estimation research.
\section*{Acknowledgements}
We acknowledge valuable discussions with Stefan Woerner, Will Zeng, Armando Bellante and Shouvanik Chakrabarti. This work is supported by the National Research Foundation, Singapore, and A*STAR under its CQT Bridging Grant. We also acknowledge funding from the Quantum Engineering Programme (QEP 2.0) under grant NRF2021-QEP2-02-P05.
\appendix
\section{Phase estimation with circular arithmetic}
In the following all angles take value in $[0,2\pi)$.
\begin{defn}[Arc] An arc $\phi\frown \varphi$ is a set of angles $\theta$ such that $|\theta\ominus \phi|+|\theta\ominus\varphi|=|\phi\ominus\varphi|$. By convention we choose the order such that $\phi\ominus\varphi<0$ i.e. the arc runs counter clockwise from $\phi$ to $\varphi$.
\end{defn}
\begin{defn}[Mapping] For $\theta \in \phi\frown \varphi$ let $\mathcal{C}_{\phi\frown\varphi}(\theta) = (\theta\ominus\phi)/(\varphi\ominus\phi)$. 
\end{defn}

It is easy to see that $\mathcal{C}$ is a bijection between $\phi\frown \varphi$ and $[0,1]$. In particular $\mathcal{C}_{\phi\frown\varphi}^{-1}(r) = (\phi + r(\varphi\ominus \phi)) \mod 2\pi$ for $r\in [0,1]$. Moreover, $\mathcal{C}$ maps the circular difference to the usual difference and vice versa i.e. $\theta_1\ominus\theta_2 = \mathcal{C}_{\phi\frown\varphi}(\theta_1)-\mathcal{C}_{\phi\frown\varphi}(\theta_2)$.

Algorithm for transforming type-II UQPE's into a low depth algorithm.
\begin{tcolorbox}
\textbf{Preprocessing}: Set $B_\text{ref} :=\epsilon_\circ, \epsilon_\text{ref} := \pi/4,  \delta_\text{ref} := \delta_\text{main}$ defined below. Run the algorithm $\text{UQPE-II}(B_\text{ref},\epsilon_\text{ref},\delta_\text{ref})$ once and let $\tilde{\theta}_\text{ref}$ be the output. Let $\phi := (\tilde{\theta}_\text{ref}- \pi/8) \mod{2\pi}$ and $\varphi := (\tilde{\theta}_\text{ref} + \pi/8) \mod{2\pi}$. 

\textbf{Main part}: Let $T= \lceil \frac{2\log(4\delta_\circ^{-1})\epsilon_\circ^{-2\beta}}{(1-r-s)^2} \rceil$, run $T$ independent $\text{UQPE-II}(B_\text{main},\epsilon_\text{main},\delta_\text{main})$ with the following parameter setting
\begin{align}
    B_\text{main} := r\epsilon_\circ,\epsilon_\text{main} := \epsilon_\circ^{1-\beta} , \delta_\text{main} := \min\bigg(\frac{\delta_\circ}{2(T+1)},\frac{s\epsilon_\circ}{4\pi}\bigg)
    \label{uqpe-ii-setting}
\end{align}
and denote the outputs by $\tilde{\theta}_1,...,\tilde{\theta}_T$.

\textbf{Postprocessing}: Let $\phi^* := (\phi- \epsilon_\text{main}) \mod{2\pi}$ and $\varphi^* = (\varphi+\epsilon_\text{main})\mod{2\pi}$. For $i=1,2,...,T$, if $\tilde{\theta}_i\notin \phi^*\frown\varphi^*$ output $\bar{\theta}=0$, otherwise let $\tilde{c}_i := \mathcal{C}_{\phi^*\frown\varphi^*}(\tilde{\theta}_i)$. Compute $\bar{c} := \frac{1}{T}\sum \tilde{c}_i$. Output $\bar{\theta}:=\mathcal{C}_{\phi^*\frown\varphi^*}^{-1}(\bar{c})$.
\end{tcolorbox}
\begin{prop}
    The above algorithm returns $\bar{\theta}$ such that $P[|\bar{\theta}\ominus\theta|\geq \epsilon_\circ]\leq \delta_\circ$
    with resources $D = \tilde{O}(\epsilon_\circ ^{-1+\beta})$ and $N = \tilde{O}(\epsilon_\circ^{-1-\beta}\log(\delta_\circ^{-1}))$.
\end{prop}
\begin{proof}
    Denote by $G_\text{pre}$ the event that $\theta \in \phi\frown \varphi$, $G_i$ the event that $|\tilde{\theta}_i\ominus \theta|\leq \epsilon_\text{main}$ for $i=1,...,T$, $G_\text{main} = \cap_{i=1}^T G_i$ and finally $G= G_\text{pre} \cap G_\text{main}$.  By construction $P[\overline{G}_\text{pre}]\leq \delta_\text{ref}$ and by the union bound $P[\overline{G}_\text{main}]\leq T\delta_\text{main}$ thus
    \begin{align}
        P[\overline{G}]\leq P[\overline{G}_\text{pre}]+[\overline{G}_\text{main}]\leq \delta_\text{ref}+T\delta_\text{main}\leq \delta_\circ/2.
    \end{align}
    Decompose each estimator in the main main into good and bad part
    \begin{align}
        \tilde{\theta}_i = \tilde{\theta}_i \mathbf{1}(G_\text{pre}\cap G_i) + \tilde{\theta}_i \mathbf{1}(\overline{G_\text{pre}\cap G_i}):= \tilde{\theta}_i^{(\text{good})}+\tilde{\theta}_i^{(\text{bad})}.
    \end{align}
    Since $P(\overline{G_\text{pre}\cap G_i})\leq \delta_\text{ref}+\delta_\text{main}\leq s\epsilon_\circ/2\pi$ we have $|\EE[\tilde{\theta}_i^{(\text{bad})}]|\leq s\epsilon_\circ$. Combined with the fact that $(\tilde{\theta}\ominus\theta) - \tilde{\theta}_i^{(\text{bad})}\leq \tilde{\theta}_i^{(\text{good})}\ominus\theta\leq (\tilde{\theta}\ominus\theta) + \tilde{\theta}_i^{(\text{bad})}$ we obtain the following bound on the bias of the good part
    \begin{align}
|\EE[\tilde{\theta}_i^{(\text{good})}\ominus\theta]|\leq |\EE[\tilde{\theta}_i\ominus\theta]|+ |\EE[\tilde{\theta}_i^{(\text{bad})}]| \leq (r+s)\epsilon_\circ.
    \end{align}
    If $G$ happens, then $\theta\in \phi\frown \varphi$ and for all $i$ we have $|\tilde{\theta}_i\ominus\theta|\leq \epsilon_\text{main}$. In other words, the true value belongs to the arc $\phi\frown \varphi$ and all estimators belong to the arc of length $2\epsilon_\text{main}$ around the true value and thus they all belong to the slightly larger arc $\phi^*\frown\varphi^*$. Thanks to the difference preserving property of the mapping, we have for all $i=1,2,...,T$
    \begin{align}
        |\tilde{c}_i-c| = |\tilde{\theta}_i\ominus\theta| \leq \epsilon_\circ^{1-\beta},\quad |\EE[\tilde{c}_i-c]| = |\EE[\tilde{\theta}_i^{(\text{good})}\ominus \theta]|\leq (r+s)\epsilon_\circ
    \end{align}
    We are now back to the domain of real numbers with the same setup as before. According to the triangle inequality 
    \begin{align}
        |\bar{c}-c|\leq |\bar{c}-\EE[\bar{c}]|+|\EE[\bar{c}]-c|\leq |\bar{c}-\EE[\bar{c}]|+(r+s)\epsilon_\circ.
    \end{align}
    By Hoeffding's inequality
    \begin{align}
        P[|\bar{c}-c|\geq \epsilon_\circ] &\leq P[|\bar{c}- \EE[\bar{c}]|\geq \epsilon_\circ(1-r-s)]\\
        & \leq 2\exp(-\frac{T\epsilon_\circ^2(1-r-s)^2}{2\epsilon^{2-2\beta}})\leq \delta_\circ/2.
    \end{align}
    Now using the difference preserving property of the inverse mapping, we conclude that
    \begin{align}
        P[|\bar{\theta}\ominus\theta|\geq \epsilon_\circ]  &= P[(|\bar{\theta}\ominus\theta|\geq \epsilon_\circ)\cap G]+P[(|\bar{\theta}\ominus\theta|\geq \epsilon_\circ)\cap\bar{G}]\\
        &\leq P[|\bar{c}-c|\geq \epsilon_\circ]+ P[\bar{G}]\leq \delta_\circ.
    \end{align}
\end{proof}
\section{Correcting Rall-Fuller's low-depth algorithm}
\label{sec: rall-fuller}
In this section we show that there is a technical error in the low-depth algorithm of Rall and Fuller \cite{Rall_2023}. We provide a modification that fixes this error.
\subsection{Overview of Rall-Fuller's original algorithm}
A polynomial $P(x)\in \mathbb{C}[x]$ is called \textit{Pellian} if there exists $Q(x)\in \mathbb{C}[x]$ such that $P$ and $Q$ are of opposite parity and 
\begin{align}
    |P(x)|^2+(1-x^2)|Q(x)|^2=1\quad \forall x\in [-1,1]\;.
\end{align}
Given $a\in [0,1]$, quantum signal processing \cite{Gily_n_2019} allows one to sample from a Bernoulli distribution of parameter $|P(a)|^2$ using $O(\deg P)$ queries if $P$ is Pellian.  Moreover, this also works if $P$ is only \textit{semi-Pellian}, which means $P$ is the real part of some Pellian polynomial. A polynomial $P$ is semi-Pellian if it has fixed parity and $|P(x)|\leq 1$ for all $x\in [-1,1]$. The idea of Rall and Fuller is to start with the confidence interval $[0,1]$ and to repeatedly shrink it by a factor of $0.9$ for $T:=\lceil\log_{0.9}\epsilon\rceil$ times until we get a confidence interval of length smaller than $\epsilon$. For each confidence interval $[a_\text{min}, a_\text{max}]$ the shrinkage is done by sampling from a Bernoulli distribution with parameter $|P(a)|^2$ for a certain number of shots where the semi-Pellian polynomial $P$ is sufficiently bounded away from $1/2$ on the left and right segment of the interval. 
\begin{figure}[ht]
    \centering
    \includegraphics[width=0.5\linewidth]{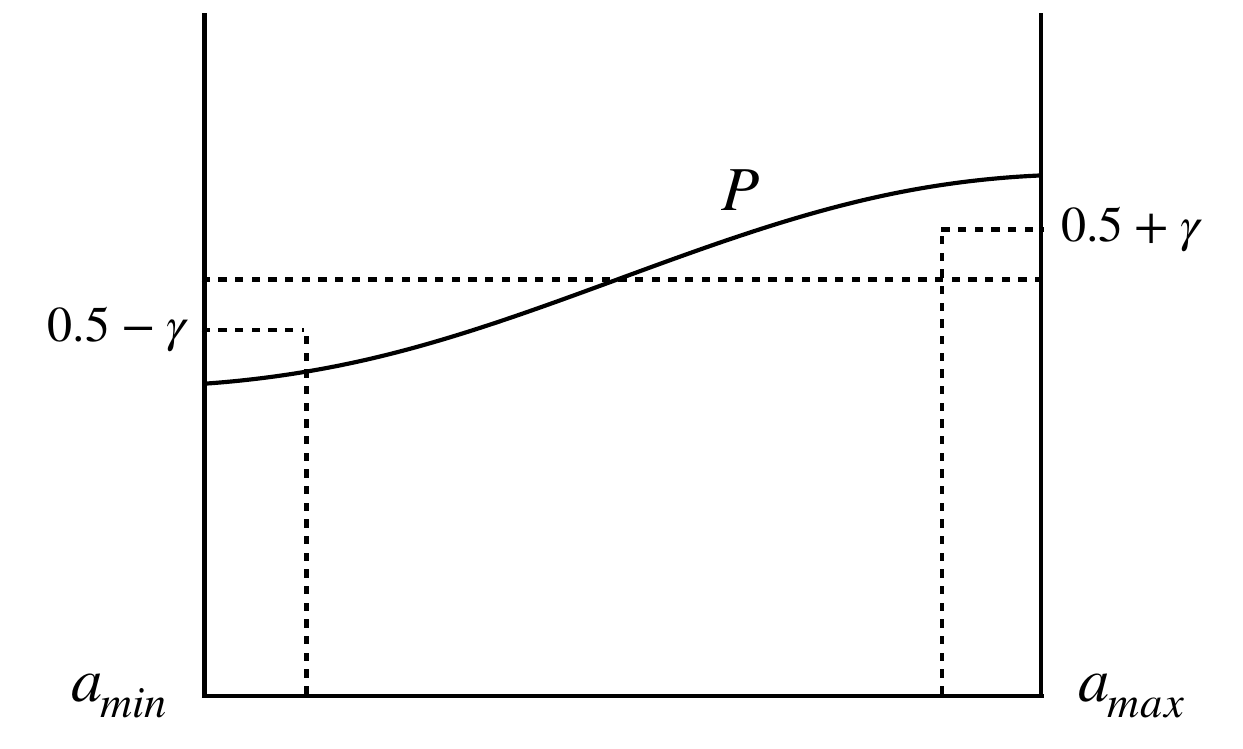}
    \caption{The desired behavior of the semi-Pellian polynomial on a confidence interval.}
    \label{fig:semi-Pellian}
\end{figure}

More precisely, let $\Delta := a_\text{max}-a_\text{min}$ then we want $P$ to satisfy
\begin{equation}
\begin{aligned}
	&P(a)\leq \frac{1}{2}-\gamma \quad \forall a\in [a_\text{min},a_\text{min}+0.1\Delta]\;,\\
	&P(a)\geq \frac{1}{2}+\gamma \quad \forall a\in [a_\text{max}-0.1\Delta,a_\text{max}]\;.
\end{aligned}
\label{desired_bound}
\end{equation}
Requiring the algorithm to be low-depth means that the degree of $P$ must be of the order of $\Delta^{-1+\beta}$. With such restriction on the degree, it turns out that the gap $\gamma$ is at best of order $O(\Delta^\beta)$. It then follows from the Chernoff bound that with $O(\gamma^{-2})=O(\Delta^{-2\beta})$ number of shots one can reject the hypothesis that $a$ belongs to the left or the right segment with high probability. The depth and the number of queries at each iteration step are thus $O(\Delta^{-1+\beta})$ and $O(\Delta^{-1-\beta})$ respectively. Since the length of the confidence interval shrinks as a geometric series, the final step dominates all and the maximum depth and the total query complexity of the algorithm are given by $O(\epsilon^{-1+\beta})$ and $O(\epsilon^{-1-\beta})$.

Now that we have understood the main idea behind Rall and Fuller's low-depth algorithm, let us fill in some technical details. The algorithm relies on two lemmas

\begin{lemma}
    Let $\tau,\eta\in (0,1)$ and $\kappa(\tau) :=\frac{1}{2}\sqrt{2\ln(2/(\pi\tau^2))}$. Consider any $k\in [1,2/\Delta]$ such that $a_\text{mid}\geq \kappa/k$ where $a_\text{mid}:=(a_\text{min}+a_\text{max})/2$. There exists an even semi-Pellian polynomial $P(a)$ such that
\begin{align}
	P(a)\leq \frac{1}{2}-0.11k(a-a_\text{mid})+\eta+0.25 \tau \quad \forall a\in [a_\text{min},a_\text{mid}]\;,\\
	P(a)\geq \frac{1}{2}+0.11k(a-a_\text{mid})-\eta-0.25 \tau \quad \forall a\in [a_\text{mid},a_\text{max}]\;.
\end{align}
The degree of this polynomial is $O(\sqrt{(k^2+\log(1/\eta))\log(1/\eta)}$.
\label{Rall-Fuller lemma 1}
\end{lemma} 
\begin{lemma}
    Let there be a coin that comes up heads with probability $x^2$. For any $\gamma, \delta >0$ let $A_{\gamma, \delta}$ be the probability of obtaining more than $\frac{1}{4}+\gamma^2$ heads in $\lceil \frac{1}{2}\gamma^{-2}\ln(\delta^{-1})\rceil$ tosses
    \begin{align}
	&\text{if}\quad x\geq \frac{1}{2}+\gamma \quad \text{then} \quad A_{\gamma,\delta}\geq 1-\delta\;,\\
	&\text{if}\quad x \leq \frac{1}{2} - \gamma \quad \text{then} \quad A_{\gamma,\delta}\leq \delta\;.
    \end{align}
\label{Rall-Fuller lemma 2}
\end{lemma} 
The first lemma constructs the desired semi-Pellian polynomial on a certain confidence interval. Basically one starts with a polynomial approximation of the error function and then applies some linear transformations along with a shift to get a polynomial $f$ with an approximate behaviour as in \cref{fig:semi-Pellian}. In order to achieve the semi-Pellian property, one needs the final result to has a well defined parity. This is done by replacing $f(x)$ with $f(x)+f(-x)$. However doing this might violate the desired gap so we want $f(-x)$ to be sufficiently small on the left and right segments. Since the un-scaled version of $f$ approximates the error function and we have shifted the origin to the mid point$a_\text{mid}$, this amounts to requiring $a_\text{mid}$ to be sufficiently large. Moreover, we want the degree of the semi-Pellian polynomial to be of order $O(\Delta^{-1+\beta})$ so we must choose $k$ to be $O(\Delta^{-1+\beta})$ as well. Then, the constraint $a_\text{mid}\geq \kappa/k$ means that such polynomial construction is only possible when the mid point is of order $\Tilde{O}(\Delta^{1-\beta})$. Note that since $a_\text{mid}= a_\text{min}+\frac{\Delta}{2}$, we are only guaranteed that $a_\text{mid}\geq \frac{\Delta}{2}$ and we thus have to consider two cases, depending whether $a_\text{mid}$ is big or small compared to $\Delta^{1-\beta}$. The full algorithm reads

\underline{\textbf{Rall-Fuller's low depth algorithm}} 
\begin{itemize}
	\item Initialize $a_\text{min}^{(0)}\leftarrow 0 $ and $a_\text{max}^{(0)}\leftarrow 1$. Let $\Delta_t = a_\text{max}^{(t)}-a_\text{min}^{(t)}$ and $a_\text{mid}^{(t)} = \frac{1}{2}(a_\text{max}^{(t)}+a_\text{min}^{(t)})$.
	\item Let $T = \lceil \log_{0.9}\epsilon\rceil$.  For $t = 0,1,...,T-1$:
	\begin{itemize}
		\item If  $a_\text{mid} ^{(t)}\geq \frac{1}{2}\Delta_t^{1-\beta}$, set:
		\begin{align}
			\eta_t,\tau_t &:= 0.01\Delta_t^\beta\;,\label{start}\\
			k_t &:=\frac{1}{2}\kappa(\tau_t)\Delta_t^{-(1-\beta)}\;,\label{k-1}\\
			\gamma_t &:= 0.01\Delta_t^{\beta}\;.
		\end{align}
		
		\item Else, set:
		\begin{align}
			\eta_t,\tau_t &:= 0.01\;,\\
			k_t &:=\frac{1}{2}\kappa(\tau_t)\Delta_t^{-1};\label{k-2}\;,\\
			\gamma_t &:= 0.01\;.\label{end}
		\end{align}
		\item Obtain $P_t(a)$ from \cref{Rall-Fuller lemma 1}. Set $\delta_t := \delta/T$, $m_t := \lceil\frac{1}{2}\gamma_t^{-2}\ln(\delta_t^{-1})\rceil$ and sample from the Bernoulli distribution with parameter $|P_t(a)|^2$ a total of $m_t$ times and let $S^{(t)}$ be the sample sum. 
		\item If $S^{(t)}\geq m_t( \frac{1}{4}+\gamma_t^2)$, set 
		\begin{align}
			[a_\text{min}^{(t+1)},a_\text{max}^{(t+1)}] : = [a_\text{min}^{(t)}+0.1\Delta_t,a_\text{max}^{(t)}] \;.
		\end{align}
		\item Else, set 
		\begin{align}
			[a_\text{min}^{(t+1)},a_\text{max}^{(t+1)}] : = [a_\text{min}^{(t)},a_\text{max}^{(t)}-0.1\Delta_t] \;.
		\end{align}
	\end{itemize}
	\item Output $\hat{a}\leftarrow a_\text{mid}^{(T)}$.
\end{itemize}
\begin{theorem}
    The output $\hat{a}$ of the above algorithm satisfies
    \begin{align}
        P[|\hat{a}-a|\geq \epsilon]\leq \delta\;.
    \end{align}
    The maximum depth and the total query complexity of the algorithm are given by
    \begin{align}
        D &= \Tilde{O}(a^{-1/(1-\beta)}+\epsilon^{-(1-\beta)})\;,\label{Rall-Fuller D}\\
	N &= \Tilde{O}((a^{-1/(1-\beta)}+\epsilon^{-(1+\beta)})\log(\delta^{-1}))\;.\label{Rall-Fuller N}
    \end{align}
    where $\Tilde{O}$ hides log factors in $\epsilon$.
\label{Rall-Fuller theorem}
\end{theorem}

\subsection{The error}
We argue that there is an error in the proof of correctness of \cref{Rall-Fuller theorem}. The authors claimed that the choice of parameters in equations \eqref{start}-\eqref{end} satisfies the constraints in \cref{Rall-Fuller lemma 1}. We show that this is not correct. There are two constraints that need to be satisfied. First, $k$ needs to be in the interval $[1,2/\Delta]$ and second, we need $a_\text{mid}\geq \kappa(\tau)/k$. We observe that the choice of $k$ in equations \eqref{k-1} and \eqref{k-2} does not satisfy the second constraint. Indeed, equation \eqref{k-2} leads to $\kappa/k=2\Delta$ and it is not always true that $a_\text{mid}\geq 2\Delta$. Instead, as mentioned above, what is always guaranteed to hold is that $a_\text{mid}\geq \frac{\Delta}{2}$. It thus seems that there is a typo in equations \eqref{k-1} and \eqref{k-2} and the authors want to set $k$ to $2\kappa(\tau_t)\Delta_t^{-1+\beta}$ and $2\kappa(\tau_t)\Delta_t^{-1}$ respectively. Even though it can be shown that this choice of parameter satisfies the second constraint, the first constraint is now violated. With $\tau_t =0.01$, $\kappa(\tau_t)\approx 2.1$ so $k = 2\kappa(\tau_t)\Delta_t^{-1}\approx 4.2\Delta^{-1}$ no longer belongs to the interval $[1,2/\Delta]$ as required.

In the following subsection, we slightly modify the parameter setting and show that the resulting semi-Pellian polynomial meets the requirement \eqref{desired_bound}.
\subsection{The new parameter setting}
For the ease of following, let us first present the explicit construction of the semi-Pellian polynomial as it was introduced in \cite{Rall_2023}.
\begin{defn}[Error function] The error function is given by
    \begin{align}
        \text{erf}(x) := \frac{2}{\sqrt{\pi}}\int_{0}^x e^{-t^2}dt\;.
    \end{align}
\end{defn}
\begin{prop}[Polynomial approximation of the error function \cite{low2017hamiltoniansimulationuniformspectral}] For all $k,\eta>0$ there exists an odd polynomial $P_{\text{erf},k,\eta}(x)$ such that for all $x\in [-2,2]$:
\begin{align}
    |P_{\text{erf},k,\eta}(x)-\text{erf}(kx)|\leq \eta\;.
\end{align}
Furthermore $\deg(P) = O(\sqrt{(k^2+\log(1/\eta))\log(1/\eta)}$.
\end{prop}
\begin{defn}[Semi-Pellian polynomial construction] For $\tau,\eta \in [0,1]$, $k>0$ and $a_\text{mid}\in [0,1/2]$, let us define
    \begin{align*}
	P_{\tau,\eta,k,a_\text{mid}}(a) &= f_0(a-a_\text{mid})+f_0(-a-a_\text{mid}),\\
	\text{where} \quad f_0(x) & := \frac{1+\eta+P_{\text{erf},k,\eta}(x)}{4\eta+\tau+2}.
\end{align*}
In the following, the dependence on $\tau,\eta,k,a_\text{mid}$ is implicitly understood and we simply denote $P_{\tau,\eta,k,a_\text{mid}}(a)$ by $P(a)$.
\end{defn}
Let us now show the following results, which constitute the proof of correctness for our new choice of parameters.
\begin{prop}
    Let $0\leq a_\text{min}<a_\text{max}\leq 1$ and $\Delta = a_\text{max}-a_\text{min}$. For $\eta=\tau=0.01$, $k=\frac{1}{2}\kappa(\tau)\Delta^{-1}$ and $\gamma=0.01$ we have
\begin{align}
	P(a)\leq 0.5-\gamma\quad \forall a\in [a_\text{min},a_\text{min}+0.1\Delta] \;,\\
	P(a)\geq 0.5+\gamma\quad \forall a\in [a_\text{max}-0.1\Delta,a_\text{max}]\;.
\end{align}
\end{prop}
\begin{proof}
     When $a$ belongs to the left interval  $[a_\text{min},a_\text{min}+0.1\Delta]$, we use the fact that $P_{\text{erf},k,\eta}(x)$ is an approximation of the error function with precision $\eta$ to obtain the following upper bound for $P$
\begin{align*}
	P(a) \leq \frac{2+4\eta+\text{erf}(k(a-a_\text{mid}))+\text{erf}(k(-a-a_\text{mid}))}{4\eta+\tau+2}\;.
\end{align*}
Note that $F(a) := \text{erf}(k(a-a_\text{mid}))+\text{erf}(k(-a-a_\text{mid}))$ is an increasing function for $a>0$ since
\begin{align*}
	F'(a) = \frac{2k}{\sqrt{\pi}}\big[e^{-k^2(a-a_\text{mid})^2}-e^{-k^2(a+a_\text{mid})^2}\big]>0\;.
\end{align*}
It follows that the upper bound achieves it's maximum value at $a=a_\text{min}+0.1\Delta$ where $F(a) = \text{erf}(-0.2\kappa(\tau))+\text{erf}(-0.3\kappa(\tau))$. We obtain with direct calculation that
\begin{align*}
	P(a)\leq \frac{2+4\eta+\text{erf}(-0.2\kappa)+\text{erf}(-0.3\kappa)}{4\eta+\tau+2}\approx 0.472<0.5-\gamma\;.
\end{align*}
 For $a\in [a_\text{max}-0.1\Delta,a_\text{max}]$, the same arguments lead to 
\begin{align*}
	P(a) \geq \frac{2+F(a)}{4\eta+\tau+2} \geq \frac{2+F(a_\text{max}-0.1\Delta)}{4\eta+\tau+2} > \frac{2+\text{erf}(0.2\kappa(\tau))-1}{4\eta+\tau+2} \approx 0.70541> 0.5+\gamma\;.
\end{align*}
where in the third inequality we have used the fact that the error function is always greater than $-1$ to get 
\begin{align}
    F(a_\text{max}-0.1\Delta) = \text{erf}(0.2\kappa(\tau))+ \text{erf}(-k(a_\text{max}-0.1\Delta+a_\text{mid}))> \text{erf}(0.2\kappa(\tau))-1\;.
\end{align}

\end{proof}
\begin{prop}
    Let $0\leq a_\text{min}<a_\text{max}\leq 1$, $\Delta = a_\text{max}-a_\text{min}$ and $\beta\in [0,1]$. For $\eta=\tau=0.01\Delta^\beta$, $k=2\kappa(\tau)\Delta^{-(1-\beta)}$ and $\gamma=0.01\Delta^\beta$, assume that $\Delta^\beta\leq 0.4$ and $a_\text{mid}\geq \frac{1}{2}\Delta^{1-\beta}$ then the conditions of \cref{Rall-Fuller lemma 1} hold, namely $k\in [1,2/\Delta]$ and $a_\text{mid}\geq \kappa(\tau)/k$.
\end{prop}
\begin{proof}
The condition $a_\text{mid}\geq \kappa(\tau)/k$ is trivially satisfied since $\kappa(\tau)/k = \frac{1}{2}\Delta^{1-\beta}$. The condition $k\leq 2/\Delta$ is equivalent to 
\begin{align}
    u(\tau):=\tau \kappa(\tau) \leq 0.01\; .
\end{align}
We will show that this always holds given that $\tau\leq 0.004$, which follows from the assumption that $\Delta^\beta\leq 0.4$. To remind $\kappa(\tau) = \frac{1}{2}\sqrt{2\ln(2/(\pi\tau^2))}$, let $\xi := - \ln \tau$ then 
\begin{align}
    \ln \kappa(\xi) = -\frac{1}{2}\ln 2 +\frac{1}{2}\ln ( \ln (2/\pi) + 2\xi)\;.
\end{align}
It follows that for $\xi>-\ln 0.004$
\begin{align}
    \frac{d\ln u}{d\xi} = -1+\frac{1}{\ln (2/\pi) + 2\xi}<0 \;.
\end{align}
That is, $u$ is a decreasing function of $\xi$ and thus an increasing function of $\tau$. Since $u(0.004)\approx 0.0092$, we conclude that $u(\tau)\leq 0.01$ for $\tau\leq 0.004$.
\end{proof}

\underline{\textbf{Modified parameter setting for Rall-Fuller's algorithm}}

\begin{itemize}
		\item If  $a_\text{mid}\geq \frac{1}{2}\Delta_t^{1-\beta}$ and $\Delta_t^{\beta}\leq 0.4$, set:
		\begin{align}
			\eta_t,\tau_t &:= 0.01\Delta_t^\beta\;,\\\
			k_t &:=2\kappa(\tau_t)\Delta_t^{-(1-\beta)}\;,\\
			\gamma_t &:= 0.01\Delta_t^{\beta}\;.
		\end{align}
		
		\item Else, set:
		\begin{align}
			\eta_t,\tau_t &:= 0.01\;,\\
			k_t &:=\frac{1}{2}\kappa(\tau_t)\Delta_t^{-1}\;,\\
			\gamma_t &:= 0.01\;.
		\end{align}
\end{itemize}
Since the newly introduced condition $\Delta_t^\beta\leq 0.4$ adds at most a constant number of queries, \cref{Rall-Fuller theorem} remains valid.
\subsection{The limitation of Rall-Fuller's algorithm}
\label{subsec: Rall-Fuller limitation}
As shown in equations \eqref{Rall-Fuller D} and \eqref{Rall-Fuller N}, the maximum depth and the total query complexity of Rall-Fuller's algorithm have a dependence on $a$. In this subsection we discuss this dependence in more detail. As mentioned previously, the low-degree semi-Pellian polynomial construction is only possible when the mid point of the confidence interval is sufficiently larger than the length of the confidence interval itself i.e. when
$a_\text{mid}^{(t)}\geq \frac{1}{2}\Delta_t^{1-\beta}$. Since $a_\text{mid}^{(t)}\geq a- \Delta_t/2$, this will be guaranteed if $a\geq \Delta_t^{1-\beta}$, that is, when $t\geq t^* := \lceil\log_{0.9}a\rceil/(1-\beta)$. The threshold $t^*$ plays a central role in Rall-Fuller's algorithm as it divides the algorithm into two phases. In the beginning, the length of the confidence interval is large compared to its  mid point, the low-degree polynomial construction is not possible and the depth scales as inverse of the confidence interval length as in a standard QAE. Then, the confidence interval shrinks small enough and we switch to the low-depth phase. Depending on the values of $a$ and $\epsilon$, it might happen that the algorithm terminates prematurely before it can transition to the second phase and thus fails to be low-depth. Since $T=\lceil \log_{0.9} \epsilon\rceil$, this happens when $\epsilon > a^{1/(1-\beta)}$  or equivalently when $D<O(1/a)$. For example, if $a=0.1$ and $\beta=1/2$ then the algorithm becomes a standard QAE for all $\epsilon> 0.01$. This threshold quickly becomes very small for larger values of $\beta$, for $\beta=3/4$, the algorithm fails to be low-depth for all $\epsilon > 10^{-4}$. Since large values of $\beta$ correspond to shallow hardware (see \cref{fig:beta}), this constitutes a bottleneck in Rall-Fuller's algorithm.
\bibliographystyle{unsrt}
\bibliography{main}
\end{document}